\documentclass[11pt]{amsart} 
\usepackage{amsaddr}

\usepackage{amsmath,amssymb}
\usepackage{amsthm}
\usepackage{physics}
\usepackage[alphabetic]{amsrefs}
\usepackage{indentfirst}
\usepackage{mathrsfs}
\usepackage{graphicx}
\usepackage{hyperref}
\usepackage{xcolor}
\usepackage{fourier}
\usepackage{geometry}
\theoremstyle{plain}
\newtheorem{theorem}{Theorem}
\newtheorem{lemma}{Lemma}

\newtheorem{prop}[theorem]{Proposition}
\newtheorem{coro}[theorem]{Corollary}

\usepackage{cleveref}

\usepackage{algorithm}
\usepackage{algpseudocode}

\theoremstyle{definition}

\theoremstyle{remark}
\newtheorem{remark}{Remark}

\newcommand{\bd}[1]{\boldsymbol{#1}}
\newcommand{\wt}[1]{\widetilde{#1}}

\newcommand{\veps}{\varepsilon}

\newcommand{\dt}{\Delta\!t}
\newcommand{\dH}{\delta\!H}
\newcommand{\dsp}{\displaystyle}

\title{A Partially Random Trotter Algorithm for Quantum Hamiltonian Simulations}
\author{Shi Jin} 
\address{School of Mathematical Sciences, Institute of Natural Sciences, MOE-LSEC and SHL-MAC, Shanghai Jiao Tong University, Shanghai, China (shijin-m@sjtu.edu.cn)}

\author{Xiantao Li}
\address{Department of Mathematics, Pennsylvania State University, University Park, PA 16802, USA (Xiantao.Li@psu.edu)}

\date{\today}
\begin{document}
\maketitle
\begin{abstract}
Given the Hamiltonian, the evaluation of unitary operators has been at the heart of many quantum algorithms. Motivated by 
existing deterministic and random methods, we present a hybrid approach, where Hamiltonians  with large amplitude are evaluated at each time step, while
the remaining terms are evaluated at random. The bound for the mean square error is obtained, together with a concentration bound. The mean square error consists of a variance term and  a bias term, arising respectively from the random sampling of the Hamiltonian terms and the operator splitting error. Leveraging on the bias/variance trade-off, the error can be minimized by balancing the two. The concentration bound provides  an estimate on the number of gates. The estimates are verified by using numerical experiments on classical computers. 
\end{abstract}

\section{Introduction}
There has been rapidly growing interest in quantum computing algorithms in quantum chemistry, where classical algorithms are often limited to small systems due to the dimension of the Fock space. At the core of many algorithms is the approximation of unitary operators \cite{cleve1998quantum,nielsen2002quantum}. Together with the phase estimation method, the unitary operators can be used to compute the ground state \cite{aspuru-guzik_simulated_2005}. 
More importantly, the ability to approximate unitary operators is crucial in the development of quantum algorithms for many other scientific-computing problems, including solutions of linear systems \cite{harrow2009quantum}, singular-value decomposition  \cite{gilyen2019quantum}, Monte-Carlo methods \cite{montanaro2015quantum}, open quantum systems \cite{childs2016efficient,cleve2016efficient}, etc.

Meanwhile, quantum algorithms have to face the challenge that the number of second quantization terms in the  Hamiltonian, which have to be treated separately,  scale unfavorably with respect to the number of orbitals. 
Existing methods, based on how the Hamiltonian terms are evaluated, can be divided into deterministic and random methods. Deterministic methods mostly rely on the Trotter-Suzuki splitting, where at each step, the unitary operator is split into the product of unitary operators, each of which only involves one Hamiltonian term. Alternatively, one can use a linear combination of unitary operators to improve the accuracy \cite{childs2012hamiltonian}. The accuracy of such methods, as well as the gate counts, have been extensively studied \cite{childs2021theory}.  Aother interesting approach, called quantum stochastic drift protocol (QDRIFT), is  first proposed by Campbell \cite{campbell2019random}, and it is stochastic in nature: At each step,  only one Hamiltonian is selected, with probability proportional to its amplitude (important sampling). It was shown that the method leads to good accuracy with high probability. Later Chen et al.  \cite{chen2020quantum} provided thorough analysis and obtained tighter error bounds and improved estimates of gate counts.

One way to compare the deterministic and random frameworks is to examine the accuracy while holding the number of gates fixed. Intuitively, due to the large number of Hamiltonian terms, deterministic methods have to be implemented with a much larger step size. Meanwhile, since only one unitary operator is evaluated at each time step, a random algorithm can be implemented with much smaller step size. However, due to the stochastic nature, the error is dominated by the variance, which is also dependent of the system size.  Motivated by these observations, we consider a hybrid algorithm, where the Hamiltonian terms with large amplitude are evaluated deterministically.  Meanwhile, the remaining terms are implemented at random. By excluding the large terms, the variance is significantly reduced. We allow the flexibility of applying multiple unitary operators at a random step.  We prove an error bound that incorporates the two error contributions. In particular, the error bound can be interpreted as the typical bias-variance trade-off in standard statistical analysis \cite{ramsay2002applied}. Therefore, the error can be minimized by balancing the two.  In addition, concentration inequalities are derived to provide an estimate of the gate count.


A similar, but separate development is the random batch methods (RBM) for simulating particle dynamics with two- and many-body interactions \cite{jin2021convergence}.  RBM has been applied to both classical \cite{jin2020random,jinLiXu2020random} and quantum systems \cite{GJP,jinli2020random}, see \cite{jin2021random}
for a recent review.  The separation of the Hamiltonian operators is motivated by the implementation of RBM in the presence of strong repulsive forces \cite{li2020random,jinLiXu2020random,jinli2020random}, in which case separating out these terms makes the time integration more accurate.  For example, in \cite{li2020random}, the inter-molecular potential is decomposed into a long range but smooth part,  and a short range potential with singularity. The long range interactions can be efficiently sampled using the RBM, followed by a Metropolis rejection step for the short range potential. In the context of quantum Monte Carlo (QMC) method, this decomposition is applied to the pairwise interactions and the two-body terms in the Jastrow factor \cite{jinli2020random}.  

Another interesting idea is to randomly order the Hamiltonian terms in the Trotter splitting algorithm. This has been implemented in \cite{tranter2018comparison,tranter2019ordering}, and recently analyzed by Childs et al. \cite{childs2019faster}. Faehrmann et al. \cite{faehrmann2021randomizing} proposed to randomize the multi-product formula to improve the sampling of an observable.  The main difference from the random algorithms considered in \cite{campbell2019random,chen2020quantum} is that the random ordering is equivalent to a random algorithm  using a uniform distribution among the Hamiltonians, rather than important sampling. In this work, we will consider both choices for the random algorithm. We will show that the important sampling generally gives a smaller variance and a lower gate count.

The rest of the paper is organized as follows. Section \ref{sec: rand} presents algorithms where some of the Hamiltonian terms are evaluated at random. We prove bounds for the mean square error in terms of the step size. In Section \ref{sec: gate}, we present concentration bounds that provide estimates for the number of gates required to achieve certain accuracy. We discuss the variance/bias balance in the mean square error in Section \ref{sec: partial}, as well as their implications to the partition of the Hamiltonians. To verify the error estimates, we present numerical experiments in Section \ref{sec: tests}.   

\section{Random Algorithms}\label{sec: rand}
\subsection{Quantum dynamics and unitary operators}
Quantum computers were originally perceived as a potential tool to solve quantum chemistry problems \cite{feynman1985quantum}. The fundamental  element behind quantum chemistry is the Hamiltonian operator, which in the second quantization form, can be written as,
\begin{equation}\label{eq: H}
 H =  \sum_{pq} h_{pq} \hat{c}_p^\dagger \hat{c}_q + \frac12\sum_{pqrs}  h_{pqrs} \hat{c}_p^\dagger \hat{c}_q^\dagger  \hat{c}_r \hat{c}_s.
\end{equation}
Here $\hat{c}^\dagger$ and $\hat{c}$ are respectively the creation and annihilation operators and they fulfill anti-commutator relations. This representation of the Hamiltonian can be obtained by a projection to atomic orbitals, although it is also possible to  employ plane wave basis \cite{babbush2018low}.

Many quantum algorithms for chemistry problems start with the observation that the dynamics of a quantum many-body problem is  determined by the unitary operator,
\begin{equation}\label{eq: ut}
 U(t)= \exp \big(-i t H\big). 
\end{equation} 
Here we have set  $\hbar=1$. We assume that the Hamiltonian  has been projected to $d$-dimensional space so $H \in \mathbb{C}^{d\times d}$. This can be 
done using the Jordan-Wigner (JW)  \cite{jordan1928pauli} or Bravyi-Kitaev (BK) approximations \cite{bravyi2002fermionic,seeley2012bravyi,hastings_improving_2014,havlivcek2017operator}, see \cite{tranter2018comparison} for more thorough discussions. For a description based on a wave function $\ket{\psi} \in  \mathbb{C}^d$ with  $\braket{\psi(0)}=1$,  the dynamics of interest is expressed as $\ket{\psi(t)}= U(t) \ket{\psi(0)}.$ On the other hand, in terms of the density matrix,
\begin{equation}
\rho(t) = \ket{\psi(t)}\bra{\psi(t)},
\end{equation}
 the time evolution is written as $\rho(t) = U(t) \rho(0) U(t)^\dagger.$ In either case, the unitary operator and its approximations play a fundamental role in understanding the underlying quantum dynamics. 

For simplicity, we assume that  the Hamiltonian has been mapped using the JW or BK method, and after such a transformation, the Hamiltonian is written as,
\[ H = h_1 + h_2 + \cdots + h_L.\]
In order to approximate the evolution corresponding to  \eqref{eq: ut}, we first divide the total Hamiltonian 
into two separate Hamiltonian operators, 
\begin{equation}\label{eq: H0H1}
 H= H_0 + H_1, 
\end{equation}
which include Hamiltonian terms as follows,
\begin{subequations}
\begin{align}
  H_1=& h_1 + h_2 + \cdots + h_{N_r}, \label{eq: H-1} \\
  H_0=& h_{N_r+1} + \cdots + h_L. \label{eq: H-0}
\end{align}
\end{subequations}

 We also assume in this partition that the number of quantum gates required to simulate the unitary operator $ \exp \big(-i \dt h_\ell \big)$ is bounded
 by a fixed constant. For details about the construction of the quantum circuits corresponding to such unitary evolution, and more detailed estimates of the required quantum resources, see \cite{whitfield2011simulation,tranter2018comparison}.
 Ideally, the separation \eqref{eq: H0H1} will result in no error if $H_0$ and $H_1$ commute. This has been explored in \cite{gokhale2019minimizing}. But 
 here we arrange the Hamiltonian terms by their magnitude in the ascending order:
\[ \|h_1\| \leq \|h_2\| \leq \cdots \leq \|h_L\|.\]
Thus the partition \eqref{eq: H0H1} corresponds to a cut-off, labelled by $N_r$. 
In the following sections, we will first focus on the treatment of $H_1$, which consists of $N_r$ Hamiltonian operators,  $1 \leq N_r \leq L,$ using a random algorithm.  In Section \ref{sec: partial}, we will study the added effect of approximating the unitary dynamics driven by $H_0.$

A standard approach to simulate  the dynamics is to divide the time interval into small time steps with step size $\dt$, and at each step, one can approximate \eqref{eq: ut} using an operator splitting scheme,
\begin{equation}
U(\dt) \approx \dsp  \exp \big(-i \dt H_0 \big) \prod_{\ell=1}^{N_r}  \exp \big(-i \dt h_\ell \big).
\end{equation}
The unitary operator $U_0:= \exp\left( -i \dt H_0\right) $ can be similarly approximated,
\begin{equation}\label{eq: U0}
U_0 \approx V_0, \quad V_0:= \prod_{\ell = N_r+1}^{L} \exp\left( -i \dt h_\ell \right).
\end{equation}
The error associated with such approximation, often referred to as the Trotter error, has been extensively studied \cite{babbush2015chemical,poulin2014trotter,hastings2014improving,tranter2019ordering,childs2019nearly}. The readers are referred to \cite{childs2021theory} and the references therein.   In particular, the error in \eqref{eq: U0} is of the order $\mathcal{O}(\dt^2)$ locally, which can be improved to  $\mathcal{O}(\dt^3)$ by using a symmetric splitting (also known as the Strang 's splitting), or higher order Suzuki's splitting \cite{seeley2012bravyi}. 
But for the moment, we neglect this splitting error, and assume $U_0$ is evaluated exactly. 


\subsection{Random algorithms}
The main challenge in directly implementing a Trotter splitting method is that the number of individual Hamiltonians in \eqref{eq: H-1} is proportional to $N^4$ with $N$ being the number of particles, i.e.,
$L= \mathcal{O}(N^4).$ To mitigate the issue, we consider an algorithm where $K$ individual Hamiltonian operators in $H_1$ are picked up in random, similar to the method by   Campbell \cite{campbell2019random}, which corresponds to the case of $K=1$. In general, the selected Hamiltonian terms form a batch with batch size $K.$  Toward this end, we introduce a Bernoulli random vector $\omega \in \{0,1\}^{N_r}$ such that, the number of non-zeros in $\omega$ is $K$; $1 \le K \ll N_r$. Specifically, let $\dt$ be the stepsize and $\{t_n\}_{n\geq 0}$ be the time steps.  At time $t_n$, we introduce the approximation,
\begin{equation}\label{eq: U-rand-ns}
U(\dt) \approx  \exp\big(-i \dt H_0 \big)  \prod_{\ell=1}^{N_r}  \exp \big(-i \frac{\dt N_r}K \omega_\ell^n h_\ell \big).
\end{equation} 

The index $n$ in the random variable $\omega_\ell^n$ indicates that they will be sampled independently at each step. The goal of this random  selection is to only evaluate $K$ Hamiltonians at a time. 
This approximation can directly be extended to a symmetric-splitting,
\begin{equation}\label{eq: U-rand-symm}
U(\dt) \approx   \prod_{\ell=1}^{N_r}  \exp \big(-i \frac{\dt N_r}{2K} \omega_{N_r-\ell+1}^n h_{N_r-\ell+1} \big) \exp\big(-i \dt H_0 \big)  \prod_{\ell=1}^{N_r} \exp \big(-i \frac{\dt  N_r}{2K} \omega_\ell^n h_\ell \big).
\end{equation} 

Consequently, the term  $\frac{N_r}{K} \sum_{\ell=1}^{N_r} \omega_\ell^n h_\ell$ can be viewed as a stochastic approximation of $H_1.$ To be more specific,
we define a random matrix,
\begin{equation}\label{eq: delH}
\delta\!H= \frac{N_r}K \sum_{\ell=1}^{N_r} \omega_\ell^n h_\ell - H_1.
\end{equation}
Clearly $\mathbb{E}[\dH]=0.$  Here we will denote the coefficients by a vector $\bd{\omega}^n =(\omega_1^n, \omega_2^n, \cdots, \omega_L^n)$. To ensure that at each step, exactly $K$
terms are selected, we enforce the condition that,
\begin{equation}
 \omega_\ell^n \in \{0,1\} \;\; \forall \ell, \;\textrm{and} \;  \dsp  \sum_{\ell=1}^{N_r} \omega_\ell^n= K.
\end{equation}
This condition will be simply written as $ \|\bd{\omega}^n\|_0=K$.

With this notation, let us first explain the evolution of the wave functions through unitary operations. We assume that $\ket{\psi(0)} \in \mathbb{C}^d$ is the initial wave function normlized such that $\braket{\psi(0)}=1$. Then the exact wave function after $n$ steps, denoted by $\ket{\psi_n}$, is given by,
\begin{equation}\label{eq: psin}
 \ket{\psi_n} = U(\dt)^n \ket{\psi(0)}, \quad n \geq 0.
\end{equation}  

On the other hand, the approximate wave function from \eqref{eq: U-rand-ns},  denoted by $\ket{\phi_n}$, is constructed as follows,
\begin{equation}\label{eq: phin}
 \ket{\phi_n} = V_{n-1} V_{n-2} \cdots V_1 V_0 \ket{\psi(0)},
\end{equation}
where for each $m \geq 0$,
\begin{equation}
 V_m = U_0(\dt) \exp\big( -i \dt (H_1 + \dH_m) \big), 
\end{equation}
with $\dH_m$ independently sampled from \eqref{eq: delH}.  
The pseudo code is illustrated in Algorithm \ref{alg: rand1}, where we used the notation $[N_r]=\{1,2,\cdots, N_r\}$.

\begin{algorithm}\label{alg: rand1}
	\caption{The Random Unitary Operations.}
	\begin{algorithmic}[1]
	\Require $\ket{\psi(0)}, N, \dt=t/N.$
		\State $U_0=\exp\big(-i\dt H_0 \big)$
	
	\For $n=1:N$
	\State Randomly pick $K$ indices $S \subset [N_r]$ 
		         \For $\; \ell \in S$
					\State $U=\exp\big(-i\dt \frac{N_r}K h_\ell \big)$
		\State $ \ket{\phi} \leftarrow U \ket{\phi} $
		\EndFor
			\State $ \ket{\phi} \leftarrow U_0 \ket{\phi} $
		\EndFor
	\end{algorithmic} 
\end{algorithm}

The choice of  $H_0$ and $H_1$ will be discussed in Section \ref{sec: partial}. 

The accuracy of operator-splitting methods \eqref{eq: U-rand-ns} can often be understood using the Baker-Campbell-Hausdorff (BCH) formula \cite{baker1905alternants,campbell1897law,hausdorff1906symbolische}, which asserts that for  any two operators $X$ and $Y$, the product of the two exponentials, $\exp(X)$ and $\exp(Y)$, can be expressed in the form of a single exponential function. Namely, 
$\exp(X)\exp(Y) = \exp(Z),$  where $Z$ can be expanded as $$Z = X+Y + \frac{1}{2}[X,Y] + \frac{1}{12}([X,X,Y]+[Y,Y,X]) + \frac{1}{24} [X,Y,Y,X] + \cdots.$$
The BCH formula provides a quick glimpse of the random algorithm. More specifically,  we  can formally write
the one-step approximation \eqref{eq: U-rand-ns} as,
\begin{equation}\label{eq: one-step}
\dsp\exp\big(-i \dt H_0 \big)  \exp \big(-i \dt (H_1 + \dH) \big) 
=\exp \Big( -i \dt (H +\dH) 
- \frac{\dt^2}{2}[H_0, H_1+\dH] 
+ \mathcal{O}(\dt^3) \Big).
\end{equation}
Since $\mathbb{E}[\delta\!H]=0$,  the leading term in \eqref{eq: one-step} can be regarded as a stochastic approximation of $H$. 

To look at the effect of this algorithm over multiple time steps, one can first extend the BCH formula to the case with multiple unitary operators. Namely we have,
\begin{equation}\label{eq: bch'}
\exp(X_1)\exp(X_2) \cdots \exp(X_n) = \exp(Z),
\end{equation}
where
\begin{equation}\label{eq: bch-Z}
Z = X_1+X_2 + \cdots + X_n + \frac12 \sum_{1\le \ell < m \le n} [X_\ell, X_m] + h.o.t.
\end{equation}
For the random algorithm \eqref{eq: phin}, if one considers the case $K=1$, then $H_1+ \dH_m= N_r h_{\ell_m}$, where $\ell_m$ is randomly drawn with probability $1/N_r$. According to \eqref{eq: bch'} and \eqref{eq: bch-Z}, we have, 
\[ \ket{ \phi_n} = \exp\big( -i tH_0 -i t N_r  \frac{1}{n} \sum_{m=1}^n  h_{\ell_m}  + \mathcal{O}(\dt) \big) \ket{\psi(0)}.\]
The term $ \frac{1}{n} \sum_{m=1}^n  h_{\ell_m},$ due to the Law of Large Numbers, converges to $H_1/N_r$. Therefore, the random algorithm can be interpreted as
a Monte Carlo method in time. A similar observation has been made in the random batch algorithm \cite{jin2021random}.

Meanwhile, the statistical error will inevitably depend on the variance. To calculate the variance, we let $\left\{\bd{u}_j \right\}_{j=1}^{N_r}$ be the standard basis vectors in $\mathbb{R}^{N_r}$.   With direct calculations, we have,
\begin{lemma}\label{var-uniform}
When $ K=1$, suppose that the random vector  $\bd{\omega}$ has probability given by,
\begin{equation}\label{eq: unif}
    \mathbb{P}(\bd{\omega}=\bd{u}_j)=\frac{1}{N_r}.
\end{equation}
 
 Then, for the random matrix $\delta \!H$, 
the mean and variance are given respectively by,
  \begin{equation}\label{eq: dH}
 \mathbb{E}[\delta\!H ]=0, \quad \Sigma:= \mathbb{E}[\delta\!H^2]= N_r \Delta, \quad \Delta:= \sum_{\ell=1}^{N_r} h_\ell^2 -\frac{1}{N_r} H_1^2.
\end{equation}
\end{lemma}

The Lemma can be extended to the case for $K>1$. With tedious calculations, we find that  $  \mathbb{E}[\delta\!H ]=0$ and
$$\Sigma:=\mathbb{E}[\delta\!H^2]=  \frac{1}{K} \frac{N_r-K}{N_r-1}  N_r \Delta.$$
So there is about a $K$ fold reduction of the variance.

 Another important extension is to sample the Hamiltonians with general discrete probability,
 \begin{equation}\label{eq: pj}
 \mathbb{P}(\bd{\omega}=\bd{u}_j)=p_j,
\end{equation}
where $p_j \geq 0$ and $\sum_j p_j=1$ represents the discrete probability. The choice of $p_j$ will be made precise later. The scenario in Lemma \ref{var-uniform} corresponds to the uniform distribution.  In general, the random algorithm is implemented by selecting a Hamiltonian
$h_j$ with probability $p_j$, and applying the unitary operator $\exp\left(-it h_j/p_j\right)$ to the wave function \cite{campbell2019random}. The corresponding variance is given by,
\begin{equation}\label{eq: Sigma}
 \Sigma = \sum_{j=1}^{N_r} \frac1{p_j} h_j^2 - H_1^2. 
\end{equation}

Throughout the paper, our estimates will depend on  some  constants which we define now. The first is the spectral norm of $\Sigma$, 
\begin{equation}\label{eq: Lambda}
 \Lambda = \|\Sigma\|.
\end{equation}
For example, for the uniform sampling in Lemma \ref{var-uniform}, we have \(\Lambda \leq N_r \sum_\ell \|h_\ell^2\|\).

In addition, we define $\Gamma$ as  the almost sure bound of $\dH$
\begin{equation}\label{eq: Gamma}
\Gamma= \sup_{\|\bd{\omega}\|_0=K} \|\dH\|.
\end{equation}
For example when $K=1$, from \eqref{eq: delH}, we can use the following simple bound,  
\begin{equation}\label{eq: Gamma'}
\Gamma= N_r\max_{1\leq \ell \leq N_r} \|h_\ell\|.
\end{equation}

 It is important to observe that both these error constants depend on $N_r$. This provides a flexibility to control various error terms. 
\medskip 

Although the interpretation using the BCH formula  \eqref{eq: one-step} provides important insights, it does not reveal an explicit error bound, an observation 
made by several previous works \cite{childs2019nearly}. In the next few sections, we will derive local and global error bounds. 

\subsection{Analysis of the one-step error}

To validate this observation, we first drop the $ \mathcal{O}(\dt^2)$ term in the exponent of \eqref{eq: one-step}, and estimate the error induced by such a stochastic approximation.


\begin{theorem}\label{eq: non-symmetric}
  Consider the exact unitary dynamics at $t=\dt$ given by the wave function, 
  \begin{equation}
 \ket{\psi(\dt)} =  \exp \big( -i \dt H \big) \ket{\psi(0)},
\end{equation}
    and let 
    \begin{equation}
\ket{\phi(\dt)} =  \exp \big( -i \dt H - i \dt  \delta\! H \big) \ket{\psi(0)}
\end{equation}
   be an approximation (with the same initial condition). Let 
   \begin{equation}\label{eq: chi}
\ket{\chi(\dt)}= \ket{\phi(\dt)} - \ket{\psi(\dt)}
\end{equation}
    be the approximation error. 
    
  Then  the error has the following bound,
  \begin{equation}\label{eq: ebound1}
 \| \chi(\dt) \| \le  \int_0^{\dt} \bra {\psi(t)} \delta\!H^2 \ket{\psi(t)}^{\frac12} dt.
\end{equation}
Here the norm is induced by the inner product. Namely,
\begin{equation}\label{eq: 2norm}
\| \chi \|:=  \braket{\chi}^{\frac12}.
\end{equation}

 Consequently, the error can be further bounded by,
\begin{equation}\label{eq: ebound1'}
\| \chi(\dt) \| \le  \dt \|\delta\!H\|.
\end{equation}
Here the matrix norm is the one induced by the vector norm \eqref{eq: 2norm}.

\end{theorem}

This bound has been proved in numerous works, e.g., \cite{chen2020quantum}. But here we will include a proof, which will be extended
to an analysis of the mean square error (MSE). Similar to the analysis in \cite{an2021time}, we will work with the wave function and measure the error by its $L^2$ norm. However, we will keep the initial state arbitrary so that the error bounds can also be carried over to the operator norms. 

\begin{proof}
We first extend the function $|\chi\rangle$ to $t\in [0,\dt]$,
 $$\ket{\chi(t)} =  \exp \big( -i t H - i t  \delta \!H \big) \ket{\psi(0)} -  \exp \big( -i t H \big) \ket{ \psi(0)}.$$
It can be directly verified that $\ket{\chi(t)}  $ satisfies the differential equation,
\begin{equation}
\frac{d}{dt} \ket{\chi(t)} = -i (H+\delta\!H)  \ket{\chi(t)} +i\delta\!H \ket{\psi(t)}. 
\end{equation}

Since $\delta\!H$ is Hermitian, we have,
\begin{equation}
\frac{d}{dt} \braket{ \chi(t)}  = i  \bra{  \chi(t)}  \delta\!H \ket{\psi(t)} - i  \bra{  \psi(t)} \delta\!H \ket{\chi(t)}=2 \text{Im} \bra{\psi(t)} \delta\!H \ket{\chi(t)}.
\end{equation}
Therefore, using Cauchy-Schwarz inequality, we have, 
\[
\frac{d}{dt} \braket{ \chi(t)}  \le 2  \braket{ \chi(t)}^{1/2}  \bra{\psi(t)} \delta\!H^2 \ket{\psi(t)}^{1/2}.
\]
By direct integration, or using the Ou-Iang inequality \cite{qin2016integral}, the  inequality leads to the desired bound \eqref{eq: ebound1}.

\end{proof}


\medskip

We now return to the random algorithm. Since the perturbation is stochastic, the error will be reflected in the variance matrix. From \eqref{eq: dH}, 
 and by applying the Jensen's inequality to \eqref{eq: ebound1}, we find that for $K=1$ with uniform sampling \eqref{eq: unif},
  \begin{equation}\label{eq: ebound2'}
 \mathbb{E} \big[ \| \chi(\dt) \| \big] \le \sqrt{N_r}  \int_0^{\dt} \bra{\psi(t)} \Sigma \ket{\psi(t)}^{\frac12} dt.
\end{equation} 

In practice, often of interest is the MSE, which can also be deduced from  \eqref{eq: ebound1} directly Using  Cauchy-Schwarz inequality:
\begin{equation}
 \braket{ \chi(\dt) }  \le  \dt \int_0^{\dt} \bra{\psi(t)} \delta\!H^2 \ket{\psi(t)} dt.
\end{equation}

The MSE immediately follows by taking the expectation and using the observation that $\ket{\psi(t)}$ is deterministic.

\begin{coro}
If $K=1$, the MSE has the following bound,
  \begin{equation}\label{eq: ebound2}
 \mathbb{E} \big[ \braket{\chi(\dt)}  \big] \le \Lambda \dt^2 .
\end{equation} 
Here $\Lambda$ is the spectral norm from \eqref{eq: Lambda}.
\end{coro}

This analysis suggests that the MSE depends critically on the variance term, which can be written as  $\mathbb{E}[\bra{\psi} \delta\!H^2\ket{\psi} ].$ 
In principle, it is possible to choose the discrete probability \eqref{eq: pj} to minimize the variance. For instance, we may construct an optimization problem,
\begin{equation}\label{eq: opt-pj1}
  \min_{ \{p_j\} } \mathbb{E}[\bra{\psi} \delta\!H^2\ket{\psi} ] - Z (\sum_j p_j -1).
\end{equation}
The last term is introduced as a Lagrange multiplier to enforce the constraint $\sum_j p_j =1.$ With direct calculations, we find that,
\begin{equation}\label{eq: pj1'}
 p_j= \displaystyle \frac{\| h_j \ket{\psi}\| } {\sum_j \|  h_j \ket{\psi}\|}.
\end{equation}

A potential drawback of this approach is that the selection depends on $\ket{\psi}$. Therefore, the probability would have to be determined at every step. A compromise is to choose
to minimize the spectral norm: 
\begin{equation}\label{eq: opt-pj2}
 \min_{ \{p_j\} } \left \|\mathbb{E}[\delta\!H^2] \right\| - Z (\sum_j p_j -1) .
\end{equation}
 However, the corresponding optimization problem does not seem to admit explicit solutions. 
An empirical extension of \eqref{eq: pj1'} is the following choice, which has been used in \cite{campbell2019random},
\begin{equation}\label{eq: pj1}
 p_j= \displaystyle \frac{\| h_j \| } {\sum_j \|  h_j \|}.
\end{equation}
With this choice, we have from \eqref{eq: Sigma} that,
\[ \Sigma = \mathbb{E}[\delta\!H^2] = \sum_j \|h_j\|  \sum_j \frac{h_j^2}{ \|h_j\|} - H_1^2.  \] 
An simple upper bound, used by \cite{chen2020quantum}, is given by, 
\begin{equation}\label{eq: Sigma'}
\Lambda = \lambda^2,  \text{with}\; \lambda= \sum_j \|h_j\|.
\end{equation}

\medskip

Now we analyze the one-step error from the non-symmetric splitting \eqref{eq: one-step}. The following estimate has been derived  in \cite{childs2019nearly}.
\begin{theorem}\label{eq: split-error}
  Consider the exact unitary dynamics given by 
  \begin{equation}
\ket{\psi(\dt)} =  \exp \big( -i \dt H \big) \ket{\psi(0)}.
\end{equation}
Suppose $H=H_0 + H_1$,    and let  $$\ket{\phi(\dt)} =  \exp \big( -i \dt H_0 \big) \exp\big(- i \dt H_1 \big) \ket{\psi(0)},$$ be an approximation  at one step using non-symmetric splitting (with the same initial condition). Let $$\ket{\chi(\dt)}= \ket{\phi(\dt)} - \ket{\psi(\dt)}, $$ be the approximation error. 
  Then the $L_2$ error has the following bound,
  \begin{equation}\label{eq: ebound'}
 \| \chi(\dt) \| \le \frac{\dt^2}2  \|[H_0,H_1]\|.
\end{equation}


\end{theorem}

More generally, one can extend Theorem \ref{eq: split-error} to include  splitting methods with multiple Hamiltonians, e.g., we consider a dynamics with $K$ Hamiltonian terms,
\begin{equation}\label{eq: multiple-splitting}
\ket{\chi} =   \exp \big( -i \dt H_1 \big) \exp\big(- i \dt H_2 \big) \cdots \exp\big(- i \dt H_K \big) \ket{\psi(0)} - \exp\big(- i \dt (H_1+H_2 + \cdots + H_K) \big)   \ket{\psi(0)}.
\end{equation}
 This can be analyzed by decomposing the error as follows,
 \[
 \begin{aligned}
 \ket{\chi} = & \exp \big( -i \dt H_1 \big)  \cdots \exp\big(- i \dt H_K \big) \ket{\psi(0)} \\ & \qquad -  \exp \big( -i \dt H_1 \big)  \cdots \exp\big(- i \dt H_{K-2} \big)  \exp\big(- i \dt (H_{K-1} + H_K) \big)  \ket{\psi(0)} \\
 &+  \exp \big( -i \dt H_1 \big)  \cdots \exp\big(- i \dt H_{K-2} \big)  \exp\big(- i \dt (H_{K-1} + H_K) \big)  \ket{\psi(0)} \\
 &\qquad - \exp \big( -i \dt H_1 \big) \cdots \exp\big(- i \dt H_{K-3} \big)  \exp\big(- i \dt (H_{K-2}+H_{K-1} + H_K) \big)  \ket{\psi(0)} \\
 & \cdots \\
 &+ \exp \big( -i \dt H_1 \big) \exp\big(- i \dt (H_{2}+\cdots + H_K) \big)  \ket{\psi(0)} \\ &\qquad - \exp\big(- i \dt (H_1+H_2 + \cdots + H_K) \big)   \ket{\psi(0)}.\\
\end{aligned}
 \]

 By applying the preceding theorem repeatedly, we arrive at the following bound,
 \begin{equation}
\| \ket{\chi}\| \leq \frac{\dt^2}2  \Big(\|[H_{K-1},H_K]\| + \|[H_{K-2},H_{K-1}]\| + \|[H_{K-2},H_{K}]\| + \cdots + \|[H_{1},H_2]\| +
 \cdots +  \|[H_{1},H_K]\| \Big).
\end{equation}
 \begin{coro}\label{cor-split-error}
 The error associated with the splitting \eqref{eq: multiple-splitting} can be bounded by,
 \begin{equation}
\|{\chi}\| \leq C K^2 \dt^2.
\end{equation}
 \end{coro}
 The constant here can be selected as $C=\frac14 \max_{1\leq j <k \leq K} \| [H_j,H_k]\|.$

\medskip

In the implementation, the unitary operator is evaluated one at a time (if $K=1)$. This amounts to approximating $H_1$ by $H_1+\delta\!H$. Therefore, it corresponds to a simple splitting as follows,
\begin{equation}\label{eq: err2}
 \ket{\chi(\dt)} =   \exp \big( -i \dt H_0 \big) \exp \big( -i \dt H_1 - i \dt \delta\!H \big)  \ket{\psi(0)} -  \exp \big( -i \dt H \big) \ket{\psi(0)}.
\end{equation}
To estimate this one-step error, we separate it into two terms, 
\begin{equation}\label{eq: chi12}
\begin{aligned}
\ket{ \chi_1(\dt) }=&  \exp \big( -i \dt H_1 - i \dt \delta\!H \big)  \ket{\psi(0)} -  \exp \big( -i \dt H_1  \big)  \ket{\psi(0)}, \\
\ket{ \chi_2(\dt) }=&     \exp \big( -i \dt H_0 \big)  \exp \big( -i \dt H_1  \big)  \ket{\psi(0)}  -  \exp \big( -i \dt H \big) \ket{\psi(0)}.
\end{aligned}
\end{equation}
As a result, we have the total one-step error,
\begin{equation}\label{eq: e1step-def}
 \ket{\chi(\dt)} =  \exp \big( -i \dt H_0 \big) \ket{\chi_1(\dt)} +  \ket{\chi_2(\dt)}.
\end{equation}

By using the triangle inequality and the estimates form Theorem 1 and Theorem 3, we arrive the following estimate:

\begin{prop}
The error \eqref{eq: e1step-def} of the one-step approximation can be bounded as follows,
\begin{equation}\label{eq: e1step}
 \braket{ \chi(\dt)}^{\frac12}  \le \int_0^{\dt} \bra{\psi(0)} U_1(t)^\dagger  | \delta\!H^2 | U_1(t) \ket{\psi(0)}^{\frac12} dt + \frac{\dt^2}2 \|[H_0,H_1]\|.
\end{equation}
Here we recall $U_1(t):=\exp \big( -i t H_1  \big) $. As a result, the error has an almost sure bound, 
\begin{equation}\label{eq: e1step-as}
 \braket{ \chi(\dt)}^{\frac12}  \le \dt  \Gamma +   \frac{\dt^2}2 \|[H_0,H_1]\|.
\end{equation}
The constant $\Gamma$ is defined in \eqref{eq: Gamma}.

 In addition, the corresponding MSE can be bounded as follows,
\begin{equation}\label{eq: e1step-mse}
 \mathbb{E}\left[ \braket{ \chi(\dt)} \right] \le 2   \Lambda \dt^2 +  {\dt^4} \|[H_0,H_1]\|^2.
\end{equation}

\end{prop}

\medskip

\begin{remark}\label{rem}
The first term in \eqref{eq: e1step} is interpreted as the statistical error, while the second term is the usual truncation error.  
It is tempting to use a symmetric splitting, also known as the Strang splitting, to improve the accuracy. In this case, the one-step error is given by,
\begin{equation}\label{eq: symm}
\begin{aligned}
  \ket{\chi(\dt)}  = & \exp \big( -i \dt/2 H_0 \big) \exp \big( -i \dt H_1 - i \dt \delta\!H \big) \exp \big( -i \dt/2 H_0 \big) \ket{ \psi(0)} -  \exp \big( -i \dt H \big) \ket{\psi(0) } \\
  = & \exp \big( -i \dt/2 H_0 \big)  \ket{\chi_1(\dt)} +  \ket{\chi_2(\dt)},   
\end{aligned}
\end{equation}
where we split the error into two error terms, given by,
\[ 
\begin{aligned}
\ket{\chi_1(\dt)} =&\Big(  \exp \big( -i \dt H_1 - i \dt \delta\!H \big)   -  \exp \big( -i \dt H_1 \big) \Big)\exp \big( -i \frac{\dt}2 H_0 \big) \ket{ \psi(0)} \\
\ket{\chi_2(\dt)} =& \exp \big( -i \dt/2 H_0 \big) \exp \big( -i \dt H_1 \big) \exp \big( -i \dt/2 H_0 \big) \ket{ \psi(0)} -  \exp \big( -i \dt H \big) \ket{\psi(0) }.
\end{aligned}\]
While $\chi_2$ can be shown to be $\mathcal{O}(\dt^3)$ \cite{wecker_gate_2014}, the statistical error from $\chi_1$ remains. 

\end{remark}

\begin{remark}

Equation \eqref{eq: e1step-as} shows a competition between the two terms, and it can be used as a guideline to choose the partition $(H_0,H_1).$ This will be discussed in Section \ref{sec: partial}.
\end{remark}

\begin{remark}
When multiple terms ($K>1$) are selected from $H_1$ at each time step, we can first generalize \eqref{eq: dH}
 to $\| \mathbb{E}[\dH^2]\| \leq  \Lambda/K$. As a result, the MSE bound in \eqref{eq: ebound2} and \eqref{eq: e1step-mse} can be improved to $2  \Lambda/K \dt^2.$ But in light of the second error term in \eqref{eq: chi12}, the second term in \eqref{eq: e1step-mse} remains unchanged. 
\end{remark}

\subsection{Analysis of the global error}
We now proceed to estimate the error over multiple steps. 
We set $\ket{\psi_0}=\ket{\phi_0} = \ket{\psi(0)}.$ The approximate wave function can be generated from \eqref{eq: phin}, by applying
\begin{equation}
 |\phi_{n+1} \rangle =  \exp \big( -i \dt H_0 \big) \exp \big( -i \dt H_1 - i \dt \delta\!H_n \big)  |\phi_n\rangle,
\end{equation}
repeatedly. Meanwhile, we let $\ket{\psi_n} = \ket{\psi(n\dt)}$ be the {\it exact} wave function, which from \eqref{eq: psin} follows a similar relation,
\begin{equation}
 |\psi_{n+1} \rangle =  \exp \big( -i \dt H \big)  |\psi_{n} \rangle.
\end{equation}
 Further, the Hamiltonians $\left\{\delta\!H_n\right\}_{n\geq0}$ are i.i.d. samples   following the statistics
\eqref{eq: dH}.

We let 
\begin{equation}\label{eq: en}
 |e_n\rangle = |\psi_{n} \rangle  -  |\phi_{n} \rangle
\end{equation}
 correspond to the numerical error. Following standard procedure for analyzing the accumulation of the error in time \cite{deuflhard2012scientific}, we can derive an error equation, given by,
\begin{equation}\label{eq: en-n+1}
 |e_{n+1}\rangle =   \exp \big( -i \dt H_0 \big) \exp \big( -i \dt H_1 - i \dt \delta\!H_n \big)  |e_n\rangle + | \chi_n \rangle,
\end{equation}
where 
\begin{equation}\label{eq: chi-n}
 \ket{\chi_n} =  \exp \big( -i \dt H \big)  |\psi_{n} \rangle -  \exp \big( -i \dt H_0 \big) \exp \big( -i \dt H_1 - i \dt \delta\!H_n \big)  |\psi_n\rangle,
\end{equation}
can be interpreted as the one-step truncation error.

Estimating the $L^2$ error is quite straightforward, since,
\[ \| e_{n+1} \| \le \| e_{n} \| + \| \chi_{n} \| \Rightarrow \|e_n \| \le \sum_{j=0}^{n-1}  \| \chi_{n} \|.\]

On the other hand, to connect to the MSE, one can start with \eqref{eq: en-n+1}. Using the Cauchy-Schwarz inequality, one finds that
\begin{equation}
\langle e_{n+1} |e_{n+1}\rangle  \le  \langle e_{n} |e_{n}\rangle + 2   \langle e_{n} |e_{n}\rangle^{\frac12} \langle \chi_n | \chi_n \rangle^{\frac12}
+ \langle \chi_n | \chi_n \rangle.
\end{equation}

Using the inequality, $2   \langle e_{n} |e_{n}\rangle^{\frac12} \le 1 +   \langle e_{n} |e_{n}\rangle,$ we arrive at,
\begin{equation}
\langle e_{n+1} |e_{n+1}\rangle  \le  \langle e_{n} |e_{n}\rangle (1 +    \langle \chi_n | \chi_n \rangle^{\frac12})
+ \langle \chi_n | \chi_n \rangle.
\end{equation}

A bound can be found by using the following  discrete Gronwall's inequality \cite{qin2016integral}[Theorem 2.1.3],
\begin{equation}
 u_{n+1} \le (1+ g_n) u_n + f_n, \forall n \ge 0  \Rightarrow u_n \le \sum_{m=0}^{n-1} f_m \prod_{j=m+1}^{n-1} (1+g_j).
\end{equation}

As a result, we have,
\begin{equation}
 \langle e_{n+1} |e_{n+1}\rangle  \le \dsp  \sum_{k=0}^n \langle \chi_k | \chi_k \rangle \exp\big (  \sum_{k=0}^n  \langle \chi_k | \chi_k \rangle^{\frac12} \big )
\end{equation}
This can be combined with the one-step error bounds \eqref{eq: e1step-as} (applied to the exponential) and  \eqref{eq: e1step-mse}.

\begin{theorem}\label{thm: mse}
The mean square of the error \eqref{eq: en} of $\ket{\phi_n}$ generated from \eqref{eq: phin} at $t=t_n$ satisfies the estimate,
\begin{equation}\label{eq: g-mse}
   \mathbb{E} \left[\braket{ e_{n} } \right]   \le  \Big(  2 \Lambda t  \dt + t \dt^3 \|[H_0,H_1]\|  \Big)
   \exp \big( \Gamma t + \frac{t\dt}2 \|[H_0,H_1]\|  \big)
\end{equation}
\end{theorem}

\begin{remark}
The two terms $2 \Lambda t  \dt + t \dt^3 \|[H_0,H_1]\|$ in the error bound can be viewed as the typical balance between the variance and the bias (also known as the bias-variance tradeoff \cite{ramsay2002applied}). In particular, the constant $\Lambda$ is from the variance $\Sigma$ \eqref{eq: Sigma} of a Monte Carlo sampling of the Hamiltonian terms, while $ t \dt^3 \|[H_0,H_1]\|$ comes from the splitting error (Theorem \ref{eq: non-symmetric}). In terms of the step size $\dt$, the bias is of much higher order.
\end{remark}

\smallskip

Finally, we can consider the general case $1 \le K \ll N_r. $ Following a similar analysis, we find that,
\begin{equation}\label{eq: msek}
   \mathbb{E} \left[\braket{ e_{n} } \right]   \le  \Big( 2 \frac{1}{K}  \Lambda t  \dt + \|[H_0,H_1]\|  t \dt^3 \Big)
   \exp \big( \frac{N_r-K}{N_r} \Gamma  t +  \frac{t\dt}2 \|[H_0,H_1]\| \big)
\end{equation}

\subsection{The mean of the random approximation}
Similar to the analysis of fully random algorithms \cite{campbell2019random,chen2020quantum},
an error bound can be obtained for the fluctuation of the random wave function. 
\begin{prop}\label{eq: mean-1step}
  Let $\ket{\phi_1} $ be the random wave function in \eqref{eq: phin}. Then, the following inequality holds almost surely,
  \begin{equation}\label{eq: err-mean}
   \left\| \ket{\psi_1}  - \mathbb{E}[\ket{\phi_1} ] \right\| \leq  \frac{\dt^2}2  \|[H_0,H_1]\| + \frac{\dt^2}2 \mathbb{E}\left[ \| (H_1+\dH) \dH^2 (H_1+\dH) \| \right]^{\frac12} .
\end{equation}  
\end{prop}

We leave the proof to the appendix. This estimate can be extended to time $t=n\dt$ \cite{chen2020quantum} using the fact that for a unitary matrix $V$, $\|\mathbb{E}[V] \|\leq 1. $ More specifically, we can split the error as follows,
\[ 
\begin{aligned}
\left\| \ket{\psi_n} - \mathbb{E}[ \ket{\phi_n}] \right\|  & \leq \left \|U(\dt) \ket{\psi_{n-1}} - \mathbb{E} [V_{n-1}] \ket{\psi_{n-1}} \right\|  + \left\| \mathbb{E} [V_{n-1}] \ket{\psi_{n-1}} -  \mathbb{E} [V_{n-1}] \mathbb{E}[ \ket{\phi_{n-1}} ] \right\|  \\  
&\leq \left \|U(\dt) \ket{\psi_{n-1}} - \mathbb{E} [V_{n-1}] \ket{\psi_{n-1}} \right\|  + \left\| \ket{\psi_{n-1}} -  \mathbb{E}[ \ket{\phi_{n-1}} ] \right\|  \\  
&\leq  ...   \\
 \Longrightarrow\left\| \ket{\psi_n} - \mathbb{E}[ \ket{\phi_n}] \right\| &\leq   \frac{t \dt}2  \|[H_0,H_1]\| + \frac{t \dt}2 \mathbb{E}\left[ \| (H_1+\dH) \dH^2 (H_1+\dH) \|\right]^{1/2}. 
\end{aligned}
 \]  

Therefore the error grows linearly and the error bound is independent of the variance $\Lambda$.

\section{Gate Counts using Error bounds in probability }\label{sec: gate}
A very useful estimate in \cite{campbell2019random} is the upper bounds for the gate counts, compared to several deterministic methods. Chen et al. \cite{chen2020quantum} used matrix concentration inequalities \cite{tropp2015introduction} and improved the gate counts. 
To obtain a gate bound for the algorithm \eqref{alg: rand1}, we first assume that $K=1$. We let $N_d=L-N_r$. In this case, since $N_d+1$ Hamiltonian terms are evaluated at each time step, we set $N_{Gate}= \mathcal{O}\big(n (N_d+1)\big)$, with $n$ being the number of time steps.

Let us first assume that in \eqref{eq: g-mse}, the stochastic error dominates. By using Markov inequality, one gets an estimate,
\begin{equation}
 \mathbb{P} \big( \|e_n\| < \epsilon \big) \geq 1 - \frac{\mathbb{E}    \left[\braket{ e_{n} } \right]  }{\epsilon^2}.
\end{equation}  

By combining with Theorem \ref{thm: mse}, we arrive at the following gate counts,
\begin{prop}
If the random algorithm \eqref{alg: rand1} with uniform sampling is applied with the number of gates,
\begin{equation}
  N_{Gate}= \mathcal{O}\left(\frac{ (N_d+1) \Lambda t^2}{ \epsilon^2 \delta}\right),
\end{equation}
then, with probability at least $1-\delta$, the solution error satisfies,
\begin{equation}
   \|e_n\| < \epsilon.
\end{equation}
\end{prop}

If both terms in \eqref{eq: g-mse} are included, one can set each term $\epsilon^2 \delta/2,$ yielding the gate count,
\begin{equation}\label{eq: gc}
 N_{Gate}= \mathcal{O}\left( \max\Big\{ \frac{ (N_d+1) \Lambda t^2}{ \epsilon^2 \delta},  (N_d +1) \sqrt[3]{\frac{ 2 t^4 \|[H_0,H_1]\| }{\epsilon^2 \delta} } \Big\}  \right).
\end{equation}

\medskip

Another bound can be obtained using the McDiarmid's inequality \cite{mcdiarmid1989method}:
\begin{lemma}
 Let $x_1, x_2, \cdots, x_n$ be i.i.d. random variables, and let $f$ be a real-valued function that satisfies the bound,
 \begin{equation}
  \left|  f(x_1, x_2, \cdots, x_i, \cdots, x_n) - f(x_1, x_2, \cdots, x_i', \cdots, x_n)  \right| \le c_i, \; \forall 1 \le i \le n. 
\end{equation}
Then, for any $\epsilon>0$, the following concentration inequality holds,
 \begin{equation}\label{eq: mcd}
 \mathbb{P}\left( \left|f - \mathbb{E}[f] \right|  > \epsilon \right)  \le  \exp \left(- \frac{2\epsilon^2}{\sum_i^n c_i^2} \right). 
\end{equation}
\end{lemma}
Various extensions of this inequality have been developed \cite{rio2013mcdiarmid,mcdiarmid1998concentration}. But here we will simply follow
\eqref{eq: mcd}. Let us first assume that the random algorithm uses the uniform sampling \eqref{eq: unif}.  
To use the McDiarmid's inequality, we observe that with the random algorithm the wave function at the $n$th step can be written in a product form \eqref{eq: phin}, where
the unitary matrices $V_m$'s are independently generated.
Next, let $\wt{V}_m$ be another realization of $V_m$, where the coefficients in  $\omega^m$ \eqref{eq: delH} with the $\ell$th entry being  1  has been replaced by $\wt{\omega}^m$ with the $\ell'$th entry being  1. 
From \eqref{eq: delH}, Theorem \ref{eq: non-symmetric}  immediately implies that,
\begin{equation}
 \| V_m - \wt{V}_m \| \le \dt N_r \big(\|h_\ell\| + \|h_{\ell'}  \|\big)  \le 2 \dt \Gamma.
\end{equation}
Here we have considered the case when the Hamiltonian is sampled according to the uniform distribution and $K=1$ so that we can use the bound \eqref{eq: Gamma'}.

By a substitution into \eqref{eq: phin},  and let $\ket{\wt{\phi}_n}=V_{n-1}  \cdots \wt{V}_{m} \cdots V_0 \ket{\psi_0}, $ we get,
\begin{equation}
 \left\| \ket{\phi_n}-  \ket{\wt{\phi}_n}  \right\| \le 2 \dt \Gamma.
\end{equation}

In accordance with the McDiarmid's inequality,  if we choose the function 
\begin{equation}
f_n=1- \text{Re} \braket{\psi_n}{ \phi_n}, 
\end{equation}
which is related to the fidelity \cite{nielsen2002quantum}. In particular, 
 we observe that the error \eqref{eq: en} is related to $f_n$ by $\braket{e_n}=2f_n.$  
With this choice for $f$, we find that $$\sum_{i=1}^n c_i^2 =4 n \dt^2 \Gamma^2. $$

Gathering these estimates we have,
\[ \displaystyle  \mathbb{P}\left(\Big|f_n - \mathbb{E}[f_n] \Big| > \epsilon\right) \leq \exp\left(- \frac{\epsilon^2}{2 n \dt^2  \Gamma^2}\right).  \]
We can set the right hand side to $\delta$, and obtain the gate count,
\begin{equation}
  N_{Gate} = \mathcal{O} \left(-  \frac{(N_d+1)\ln(\delta) t^2   \Gamma^2 }{\epsilon^2} \right),
\end{equation}
which guarantees that $\Big|f_n - \mathbb{E}[f_n] \Big| < \epsilon$ with probability at least $1-\delta.$ This gate count involves $\Gamma$, which scales linearly with $N_r.$ But
this estimate is rather crude. In particular, the bounded differences used an almost sure bound. This can be improved by using an extension of the McDiamid's inequality \cite{kutin2002extensions},
where the bounded differences only need to hold with large probability.

\medskip

Let us now turn to the case with the important sampling \eqref{eq: pj1}. Using the bound \eqref{eq: Sigma'}, we have, 
\[ V_m= \exp\left(-it \frac{h_\ell}{p_\ell} \right), \quad p_\ell= \frac{\|h_\ell\|}{\sqrt{\Lambda}}, \]
which implies that the bounded differences are given by,
\begin{equation}
 \left\| \ket{\phi_n}-  \ket{\wt{\phi}_n}  \right\|  \le 2 \dt \sqrt{\Lambda}.
\end{equation}
Therefore, we arrive at a gate count estimate where the almost sure bound $\Gamma$ is replaced by $\Lambda$ \eqref{eq: Lambda}. The estimate is similar to  those obtained in \cite{chen2020quantum}.


\begin{theorem}\label{thm: gate-counts}
For any $\delta >0,$ if the random algorithm \eqref{eq: phin} is implemented with important sampling with the number of gates given by,
\begin{equation}
  N_{Gate} = \mathcal{O} \left(-  \frac{(N_d+1)\ln(\delta) t^2   \Lambda }{4\epsilon^2} \right),
\end{equation}
then the follow inequality holds,
\begin{equation}
 \mathbb{P}\left( \big|f_n - \mathbb{E}[f_n] \big|  < \epsilon \right) > 1 - \delta.
\end{equation}

\end{theorem}
Again we have assume that the variance term in the MSE is dominant. One can also include both terms and derive a gate estimate similar to \eqref{eq: gc}.
This theorem provides an estimate of the number of gates that is needed in order for the error to be within certain threshold $\epsilon$ with high probability.

\section{Partially Random Algorithms}\label{sec: partial}
Motivated by the analysis from previous sections, one can see that the random algorithm has an MSE that is proportional to $ \Lambda \dt,$  which can be rather large when some of the Hamiltonian terms have large magnitude. On the other hand, a deterministic algorithm, e.g., the ones based on direct operator splitting, has an error at most proportional to $\dt^2$ \cite{childs2021theory}. This may suggest that the latter approach is more accurate. But the comparison should be made within the same computational complexity, e.g., on the ground that one is using the same number of gates. Here we denote it by $N_{Gates}.$ For instance, for a fully random algorithm, where only one Hamiltonian is selected at a step, we have $\dt = T/N_{Gates}$. Meanwhile, for a fully deterministic method, all $L$ Hamiltonians are evaluated, which gives, $ \dt =L T/N_{Gates}.$ Therefore, for a large system with many terms in the Hamiltonian, the step size is rather large, which leads to large error. 

In light of the tradeoff, we will consider partially random algorithms, which can be considered as a hybrid of the two. Recall that $ N_d = L- N_r$. We then partition the total Hamiltonian according to \eqref{eq: H-0} and \eqref{eq: H-1}.
The unitary operator $U_0= \exp\left( -i \dt H_0\right) $ is treated using an operator-splitting \eqref{eq: U0}. 
Meanwhile, the remaining Hamiltonians, with relatively smaller magnitude, will be treated randomly. At each step, one picks up $N_d$ unitary operator from \eqref{eq: U0} and $K$ unitary operators from  $H_1$ in \eqref{eq: H-1}. This suggests that we choose the step size as follows,
\begin{equation}\label{eq: dt}
    \dt = (N_d+K) T/N_{Gate}.
\end{equation}

Intuitively, choosing a larger $N_d$ will reduce the variance $\Lambda$, at the expense of increasing the splitting error.  To understand the error of the hybrid method, let us first recall $U=\exp\left(-i\dt H\right)$ and $U_1=\exp\left(-i\dt H_1\right)$, alongside with its random approximation $V_1=\exp\left(-i\dt h_\ell/p_\ell \right),$ with probability $p_\ell,$ $ 
\ell=1,2, \cdots, N_r$.  Therefore, the one-step  error can be decomposed as follows,
\begin{equation}
  \ket{\chi} = V_0 V_1 \ket{\psi}  - U \ket{\psi} =: \ket{\chi_1} + \ket{\chi_2}, 
\end{equation}
where we have defined,
\begin{equation}
\begin{aligned}
  \ket{\chi_1}&= V_0 V_1 \ket{\psi}  - V_0 U_1 \ket{\psi}, \\
    \ket{\chi_2}&= V_0 U_1 \ket{\psi}  -  U \ket{\psi}.
\end{aligned}
\end{equation}

Due to the fact that $V_0$ in \eqref{eq: U0} is unitary, the MSE from $\ket{\chi_1}$ follows from the estimate  \eqref{eq: e1step-mse}. Meanwhile, the approximation in $\ket{\chi_2}$ is a standard operator-splitting, and in light of  Corollary \ref{cor-split-error}, the square error is given by,
\[
   \braket{\chi_2} \leq C \dt^4.\]
   Here, following the BCH formula, we will choose 
\begin{equation}\label{eq: cnst}
C= \| Q^2 \|/4, \quad Q:=[H_1,H_0]+\sum_{N_r +1 \leq j<l \leq L} [h_j,h_\ell].
\end{equation}

We now examine the cross term, $\braket{\chi_1}{\chi_2}.$   We notice that $\ket{\chi_2}$ is deterministic and it is of  order $\dt^2.$ On the other hand, according to 
Proposition 7, the mean of $\chi_1$ is also of  order $\dt^2.$ Therefore, the cross term is also of the order $\dt^4,$ and it can be absorbed into the term $\braket{\chi_2}$ term.  Based on the above heuristic analysis, we consider the following bound, 
\begin{equation}\label{eq: err-estimator}
  MSE \le \Lambda \dt^2  + C \dt^4.
  \end{equation}
For a fully random algorithm, we have $C=0$; But for a deterministic method,  $\Lambda=0,$ suggesting that a balance might be struck between the two to minimize the total MSE. It is helpful to look at the error while fixing a time $T$ and the number of gates $N_{Gate}$. Combining \eqref{eq: dt} and \eqref{eq: err-estimator}, we have the MSE at time $t=T,$
\begin{equation}\label{eq: err-Nd}
    MSE(T) \leq \frac{\Lambda(N_d) \big(N_d+K\big) T^2}{N_{Gates}} + \frac{ C(N_d+K)^3 T^4}  {N_{Gates}^3}.
\end{equation}
Here we have explicitly indicated the dependence of $\Lambda$ on $N_d$. Therefore, the optimal partition \eqref{eq: H0H1} amounts to finding $N_d$ so that the above error is minimized.

\section{Numerical Tests}\label{sec: tests}

\subsection{The mean square error}
We conducted a number of numerical experiments to test the MSE bounds in Theorem \ref{thm: mse} and equation \eqref{eq: msek}.  
In these experiments, we start with $\ket{\psi(0)} $ as the ground state with the ground state energy $\veps,$
with Hartree as the unit.  The exact solution is given by $$\ket{\psi(t)}= \exp(-i \veps t) \ket{\psi(0)}.$$
Meanwhile, the approximate solution, denoted by $\ket{\phi_n}$ is generated from \eqref{eq: phin}.
To compute the MSE,  we use 80 ensembles. All simulations are performed on classical computers using MATLAB.

As a test example, we consider the methane molecule. The Hamiltonian is obtained from OpenFermions \cite{mcclean2020openfermion}.  In the implementation, we use the Jordan-Wigner transformation \cite{jordan1928pauli,whitfield2011simulation} 
 to represent each operator as a linear transformation on $\mathbb{C}^{2^n}$.    After combining each Hamiltonian term with its transpose, we have 179 quartic terms ($L=179$). On the other hand, the quadratic terms are lumped into $H_0$. To examine the error bounds, we implement $U_0$ \eqref{eq: U0} exactly to focus on the stochastic error and the splitting error between $H_0$ and $H_1$. The ground state energy estimated from this Hamiltonian is $-36.8137$ Hartree.  The random algorithm is implemented up to time $t=10.$ In our first test, we examine the MSE for various choices of the step size $\dt$. As show in Figure \ref{fig: methane}, a smaller step size generally yields a smaller error. For  short time, e.g., $t\in [0,2]$, the error exhibits a linear dependence on $\dt$. For larger $t$, this dependence seems more subtle, which can be attributed to the exponential term in the error estimate \eqref{eq: g-mse}.  

\begin{figure}[htbp]
\begin{center}
\includegraphics[scale=0.22]{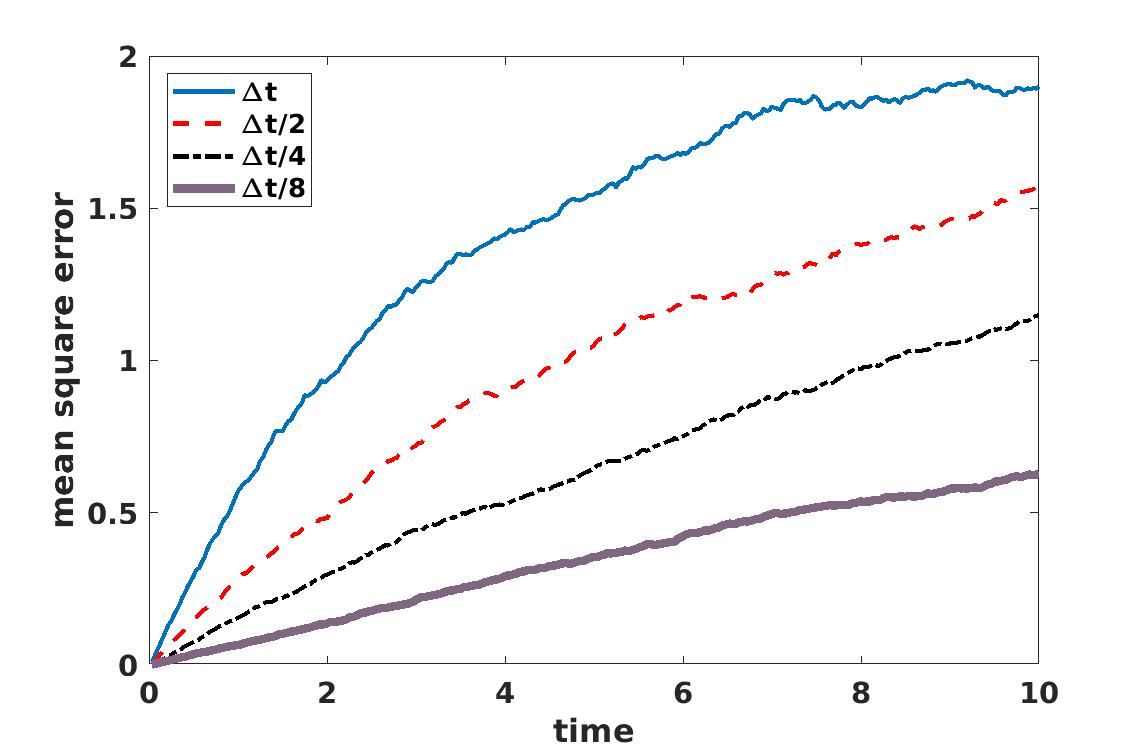}
\caption{The mean square error for the solution obtained from the random algorithm \eqref{eq: phin} computed for the methane molecule. These tests started with $\dt=0.0125$,  and then repeated by reducing the step size by a factor of 2 at each time.}
\label{fig: methane}
\end{center}
\end{figure}

To further investigate the order of the error, we show in Figure \ref{fig: msefit} the MSE at $t=1.25$, which  fits well with  a linear profile. This further verifies the first-order accuracy obtained from the analysis. To further verify the prefactor, we choose $K=1$ and $K=10.$ Interestingly, the choice $K=10$ does reduce the prefactor by a factor close to 10 as suggested by \eqref{eq: msek}. 

\begin{figure}[htbp]
\begin{center}
\includegraphics[scale=0.22]{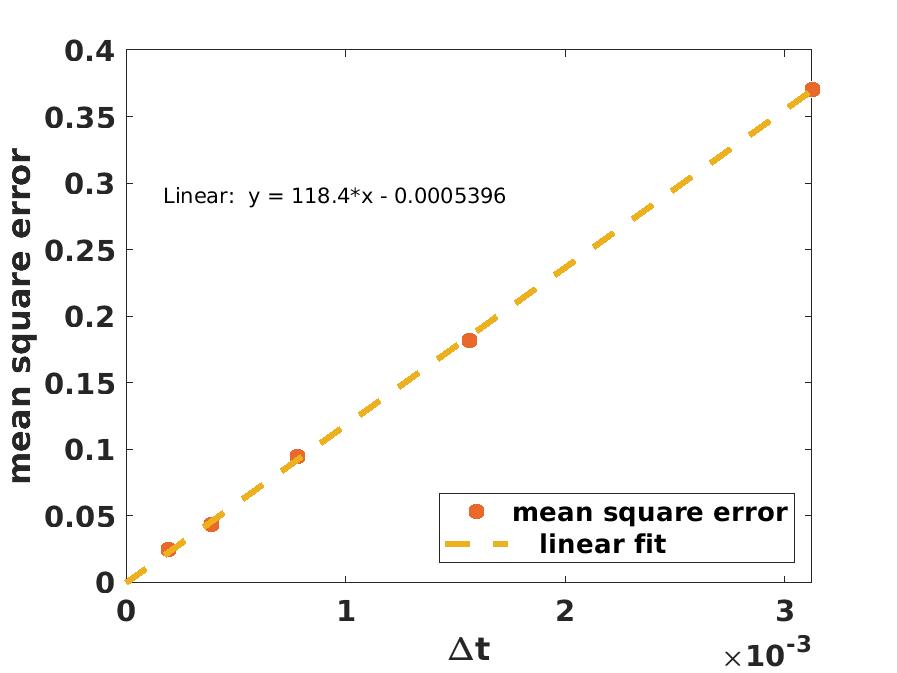}\quad
\includegraphics[scale=0.23]{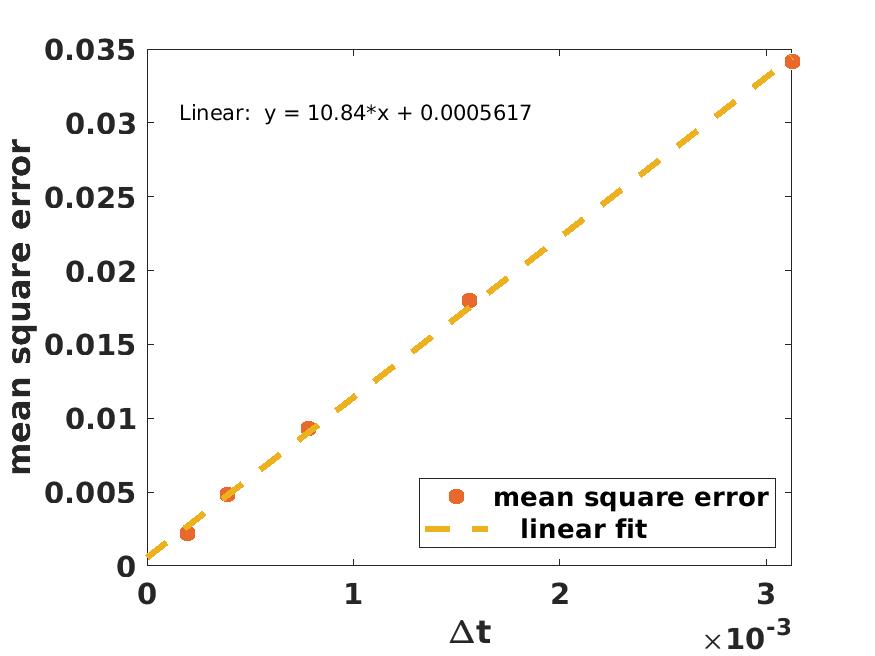}
\caption{The mean square error for various choices of the step size $\dt$. Left panel: $K=1$; Right panel: $K=10$.}
\label{fig: msefit}
\end{center}
\end{figure}

\subsection{Partially random algorithms}
\quad\\

In this section, we present results from several numerical tests  to examine the bias-variance trade-off in \eqref{eq: err-estimator}. We first choose the initial condition $\ket{\psi(0)}$
as a linear combination of the first 100 eigen modes with coefficients randomly chosen so that the observation is not specific to a particular initial state, e.g., the ground state.  

\smallskip

\noindent{\bf A Methane molecule.} We first consider a Methane molecule. The total Hamiltonian consists of 185 terms, with each Hamiltonian $h_\ell \in \mathbb{C}^{64\times 64}$. The magnitude, $c_\ell = \|h_\ell \|$, is displayed, in descending order, in Fig. \ref{fig: meth_coeff}. Clearly there are significant drops in the magnitude.
\begin{figure}[htbp]
\begin{center}
\includegraphics[scale=0.205]{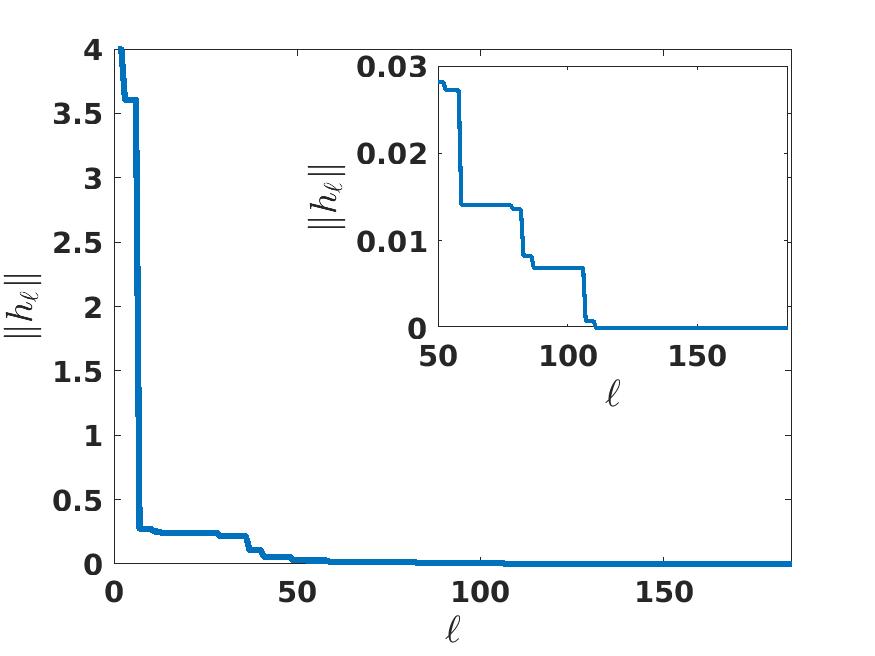}
\caption{The magnitude of the Hamiltonians for the methane molecule. The inset shows the coefficients $c_\ell $ for $\ell \geq 50.$} 
\label{fig: meth_coeff}
\end{center}
\end{figure}

To understand the MSE from the hybrid scheme, we set $N_{Gate}=2^{10},$ and ran simulations with $N_d=0, 10, 20, \cdots.$ The corresponding MSE within
the time period $[0, 4]$ is shown in the left panel of Fig. \ref{fig: methane_ng1024}. We observe that in this case, the fully random method has the largest error, and the error decreases when $N_d$ is increased. We repeat the experiment over a much larger interval $[0, 40]$, and the MSE is shown in the right panel of Fig. \ref{fig: methane_ng1024}. Due to the larger time interval, and because of the fixed number of gates, the fully deterministic algorithm is forced to use larger step size, tilting the balance between the bias and variance in \eqref{eq: err-estimator}. Remarkably, for some intermediate values of $N_d$ (around $N_d=100$), the hybrid algorithm achieves the best accuracy.  This could be linked to the abrupt change of the coefficients at that location, as shown in Fig. \ref{fig: meth_coeff}.
\begin{figure}[htbp]
\begin{center}
\includegraphics[scale=0.117]{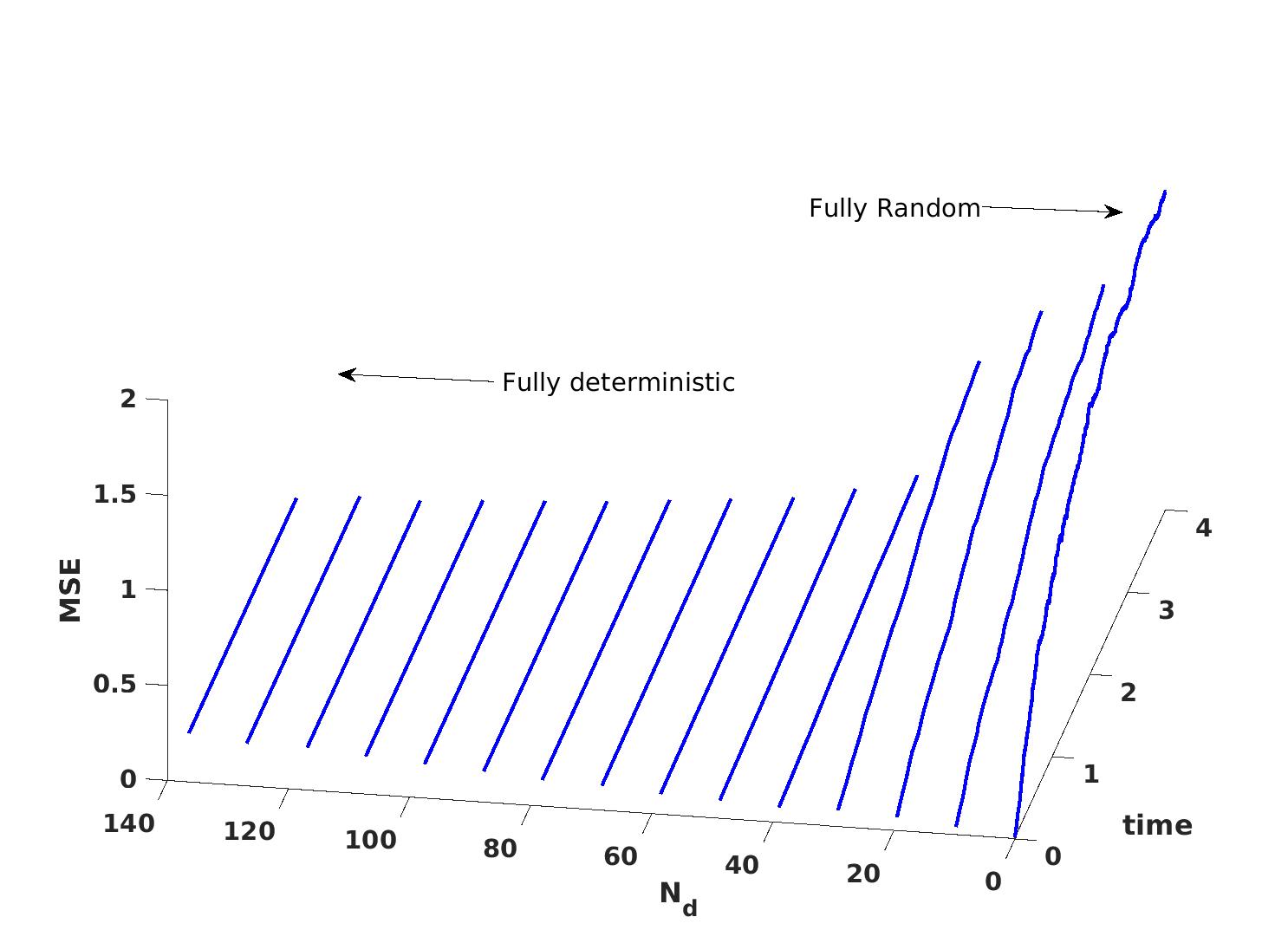}
\includegraphics[scale=0.12]{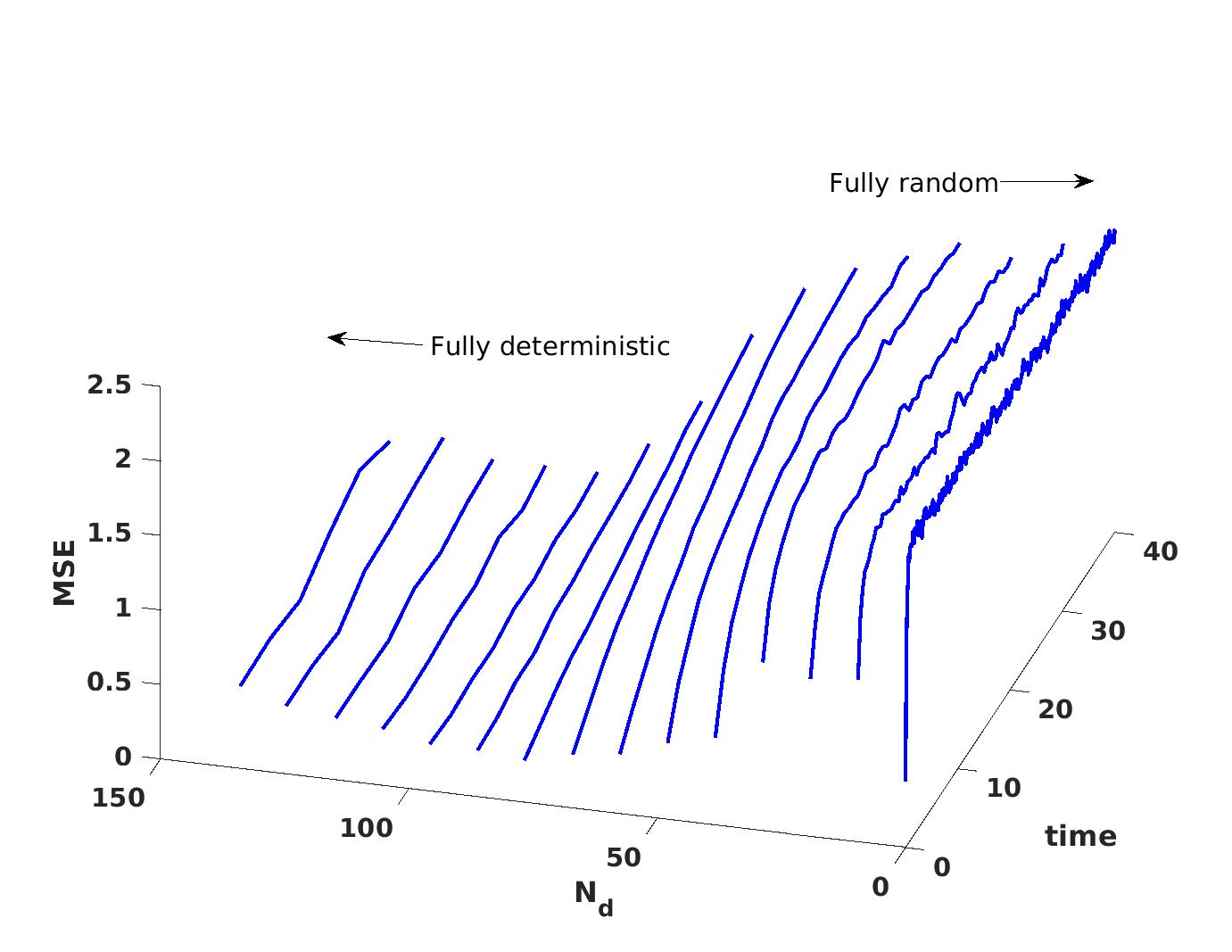}
\caption{The mean square error for the hybrid algorithm applied to a Methane molecule. Left: Error for $t \in [0, 4]$. Right: Error in the window
$[0,40]$. }
\label{fig: methane_ng1024}
\end{center}
\end{figure}

\medskip

\noindent{\bf A Dimethylamine molecule.} Next we consider a much larger molecule.  In this case, each Hamiltonian $h_\ell \in \mathbb{C}^{4096\times 4096}$, and by removing the Hamiltonians with coefficients $c_\ell < 10^{-4},$ we obtain $L=1306$ terms. The coefficients are shown in Fig. \ref{fig: dimmeth_coeff}.
\begin{figure}[htbp]
\begin{center}
\includegraphics[scale=0.225]{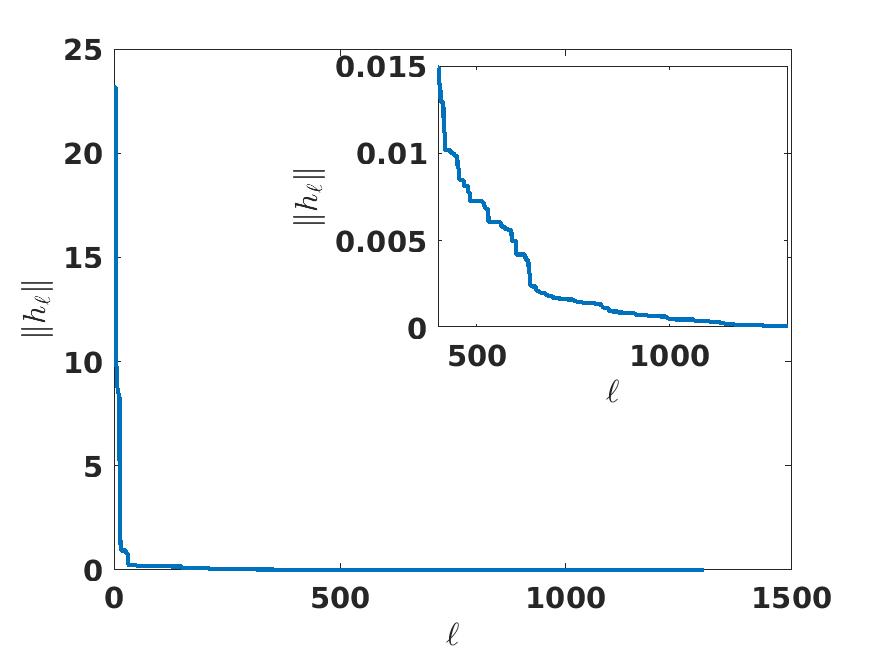}
\caption{The magnitude of the Hamiltonians for the Dimethylamine molecule. The inset figure shows the values of $c_\ell$ for $\ell \geq 500$. }
\label{fig: dimmeth_coeff}
\end{center}
\end{figure}

First we set  $N_{Gate}=2^{12},$ and ran simulations with $N_d=0, 50, 100, \cdots.$  The MSE is shown in the left panel of Fig. \ref{fig: dimeth}.
In this case, the fully random algorithm seems to give the largest error. The fully deterministic method turns out to be much better. But the partially random algorithm with $N_d=750$ gives the best accuracy. We repeated the experiment by increasing the number of gates to $N_{Gate}=2^{13}.$ In this case, the best case is when 
$N_d=650.$ See the right panel in Fig. \ref{fig: dimeth}
\begin{figure}[htbp]
\begin{center}
\includegraphics[scale=0.15]{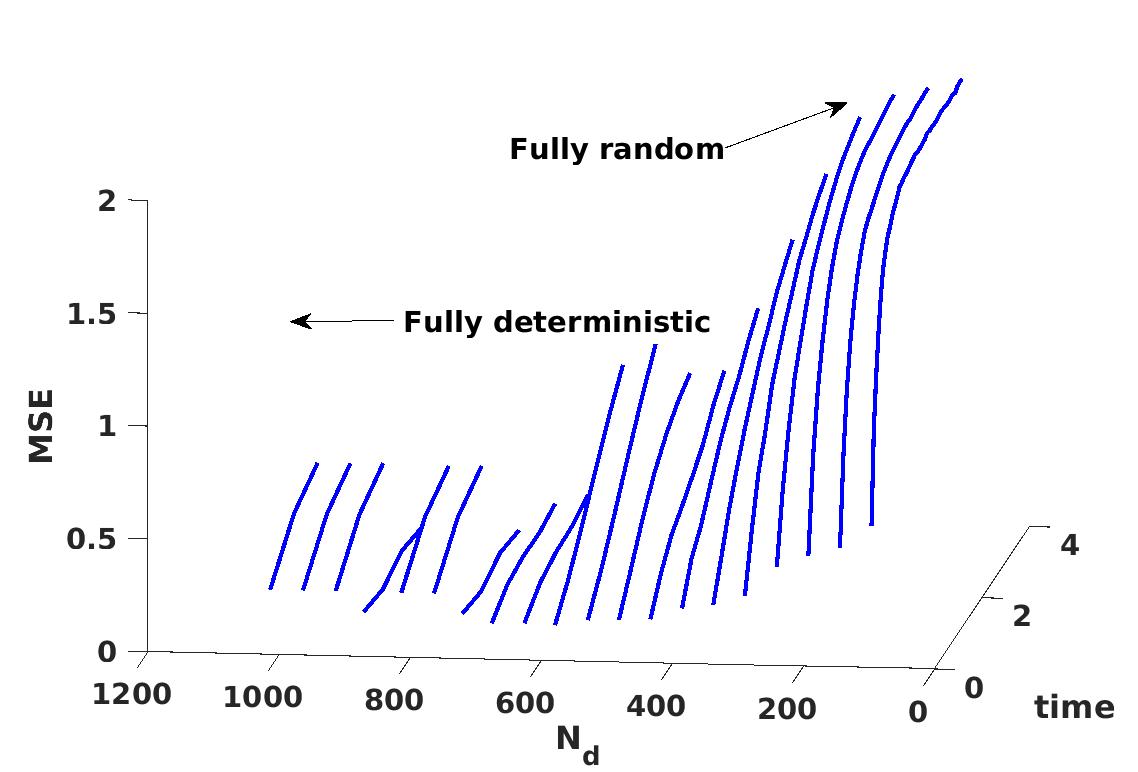}
\includegraphics[scale=0.15]{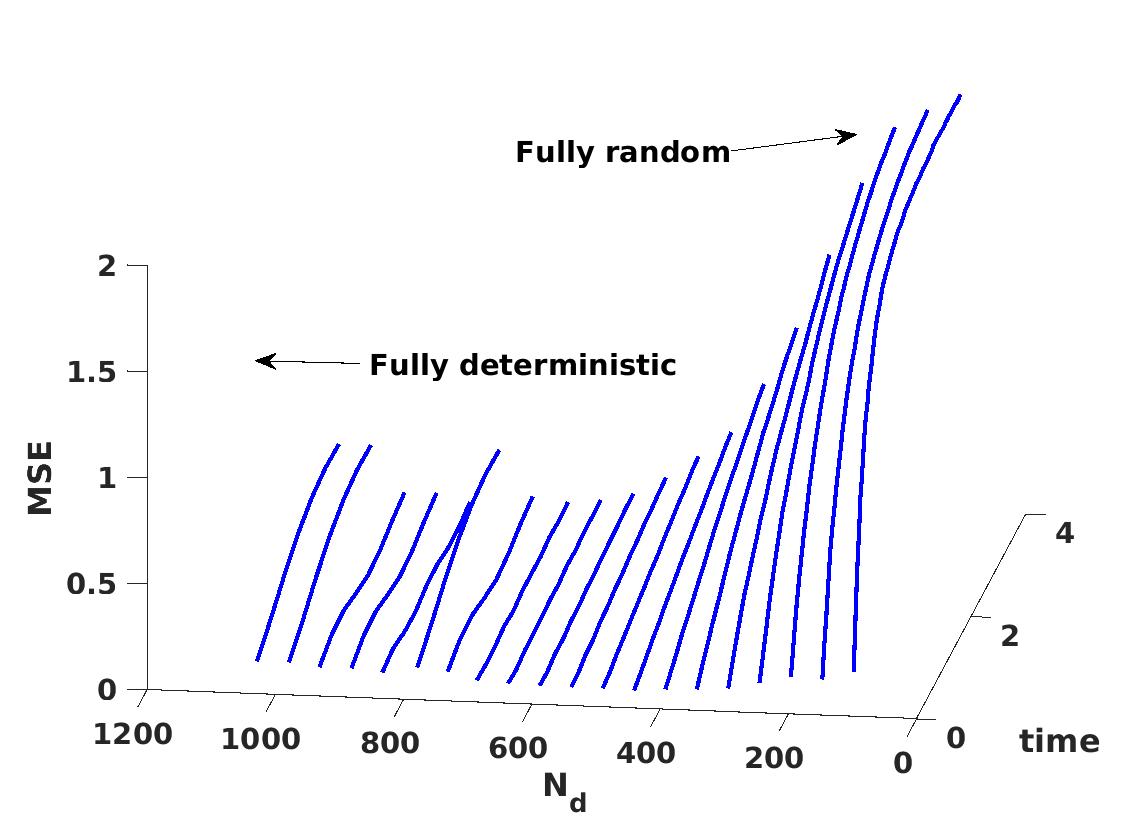}
\caption{The mean square error for the hybrid algorithm for the dimethylamine molecule. Left: $N_{Gate}=2^{12};$ Right: $N_{Gate}=2^{13}.$}
\label{fig: dimeth}
\end{center}
\end{figure}

\smallskip

\noindent{\bf Heisenberg Chains with nonlocal interactions.} 
Our third example is the Heisenberg chain model from \cite{tran2019locality}, with total Hamiltonian given by,
\begin{equation}\label{eq: hchain}
 H = \sum_{i=1}^{n-1} \sum_{j=i+1}^n \frac{1}{|j-i|^4} \big( \sigma_i^x  \sigma_j^x + \sigma_i^y  \sigma_j^y  + \sigma_i^z  \sigma_j^z \big) +
 \sum_{i=1}^n B_i \sigma_i^z.
\end{equation}
Here $n$ refers to the length of the chain. The coefficients $B_i$ are randomly chosen from $[-1,1]$ and held fixed afterwards. By treating the Pauli matrices $\sigma^x, \sigma^y$, and $\sigma^z$, separately in the Hamiltonian, the number of terms is $L=\frac32 n^2 - \frac12 n.$ The coefficients $|i-j|^{-4}$ in the Hamiltonian represents a power-law interaction \cite{tran2019locality}. The norms of the Hamiltonian terms, in descending order, is shown in Fig. \ref{fig: chain_coeff} for the two cases $n=10$ and $n=12.$ Notice that the dimension of each Hamiltonian term $h_\ell$ is given by $2^n \times 2^n.$ Clearly, there are sudden drops in the magnitude. 

\begin{figure}[htbp]
\begin{center}
\includegraphics[scale=0.205]{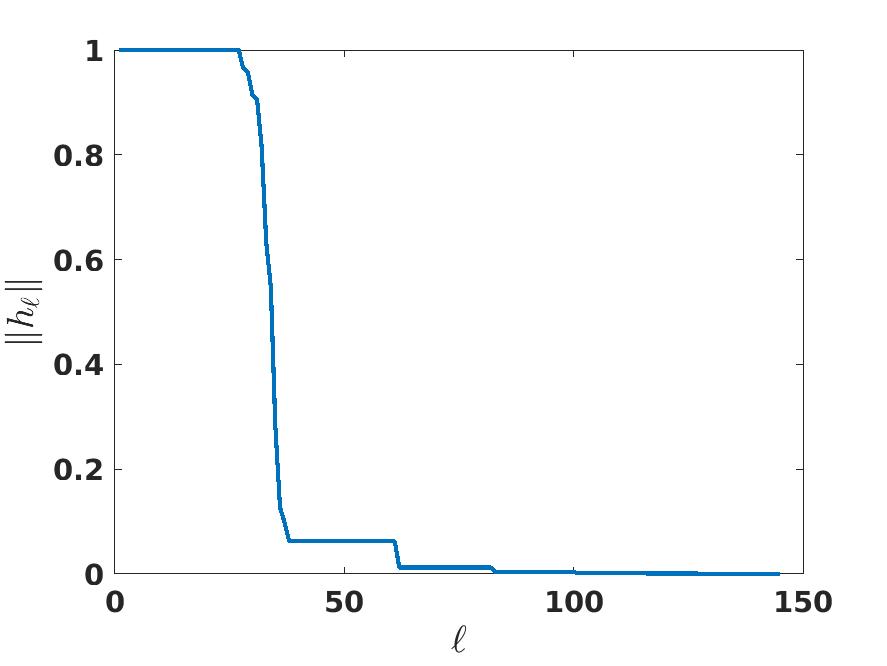}
\includegraphics[scale=0.205]{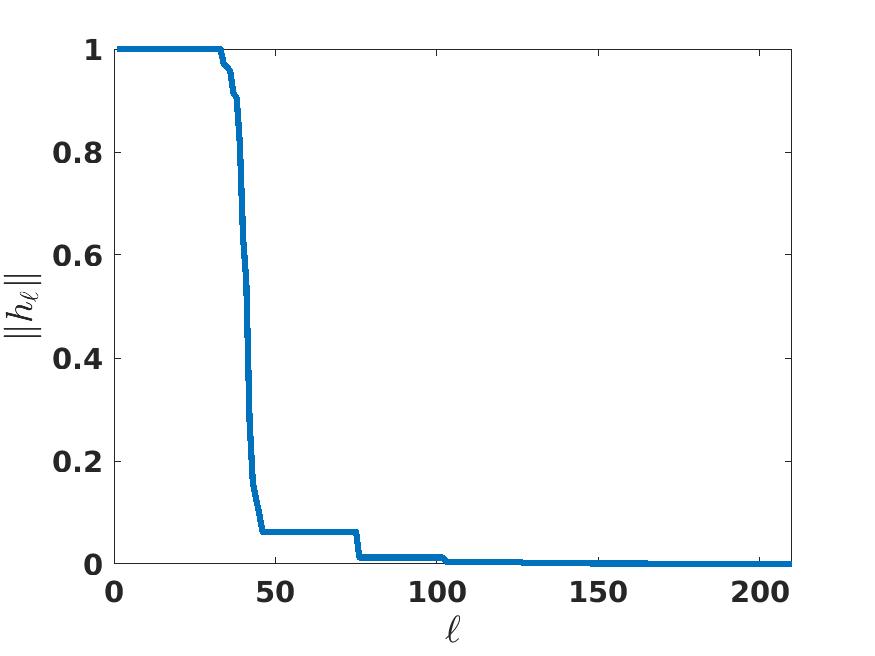}
\caption{The magnitude of the Hamiltonians for the Heisenberg chain \eqref{eq: hchain} in descending order. Left: $n=10$; Right: $n=12.$ The batch size $K=1.$}
\label{fig: chain_coeff}
\end{center}
\end{figure}

We have run the partially random algorithms with $N_d=0, 10, 20, \cdots.$  The MSE is shown in Fig. \ref{fig: chain}. In both cases, we observe that the fully deterministic and the 
fully random algorithms both have the worst accuracy.  Instead, a hybrid method with $N_d \approx 50$ yields the best accuracy. 

\begin{figure}[htbp]
\begin{center}
\includegraphics[scale=0.15]{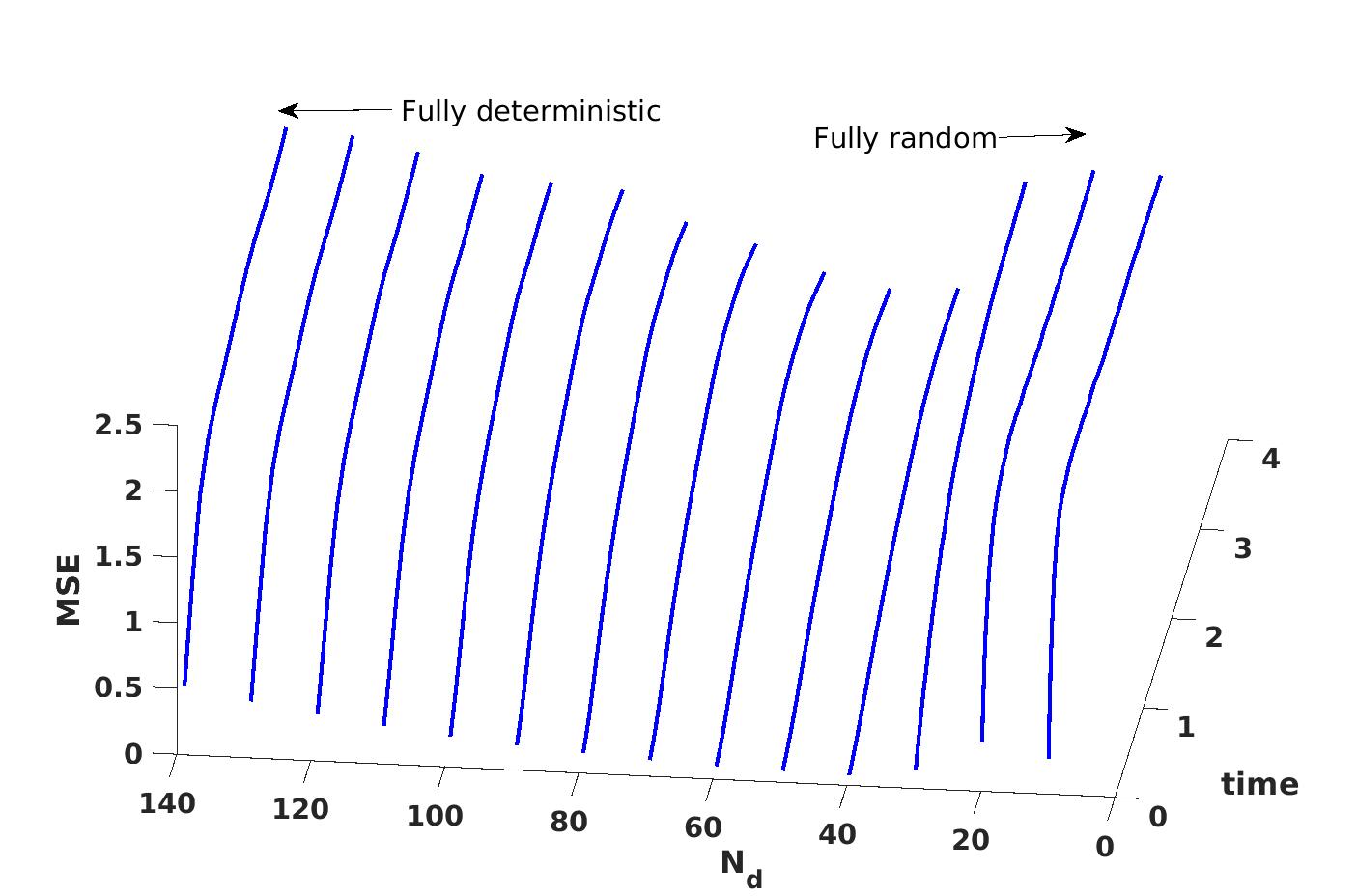}
\includegraphics[scale=0.1333]{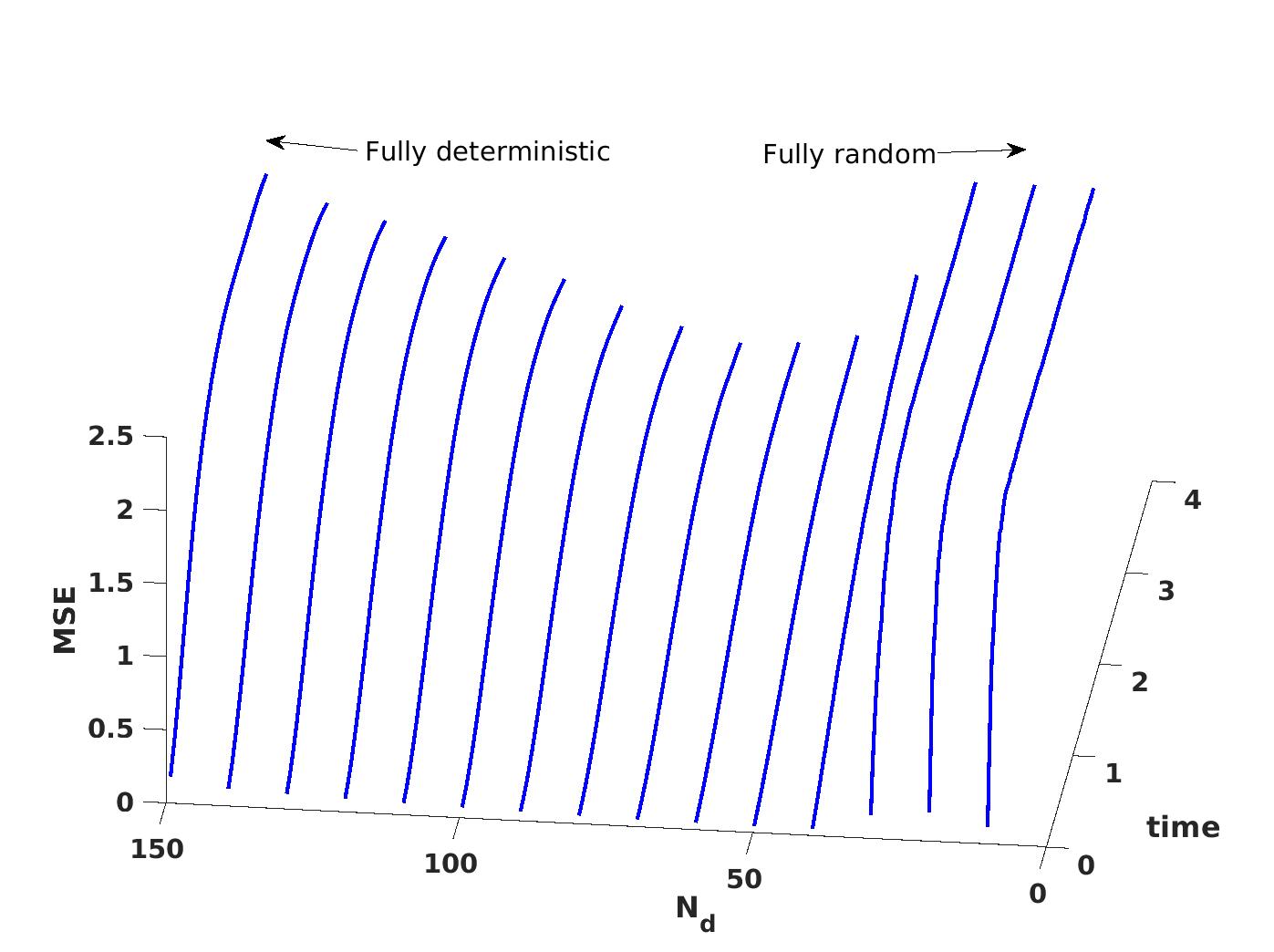}
\caption{The MSE for the hybrid algorithm applied to the model \eqref{eq: hchain}. Left: $n=10$ and $N_{Gate}= 2^{11}$; Right: $n=12$ and $N_{Gate}= 2^{12}$. The batch size $K=1.$}
\label{fig: chain}
\end{center}
\end{figure}

We now turn to the error estimator \eqref{eq: err-estimator} to understand how the error behaves with different choices of $N_d$. For this purpose, we 
consider the test case $n=12$  and $N_{Gate}= 2^{12}$ (see the right panel of Fig. \ref{fig: chain}). Specifically, we evaluate the two terms on the right hand side of \eqref{eq: err-estimator}. For each choice of $N_d,$ the variance $\Lambda$ is computed from \eqref{eq: Sigma} and  \eqref{eq: Lambda}. The probability distribution is 
defined based on the norms of $h_\ell$, from  \eqref{eq: pj1}. In addition, the constant $C$ is evaluated based on \eqref{eq: cnst}. Since different choices of $N_d$ corresponds to different choices of $\dt$, we compare the error at the same instance $t=0.2,$ by adding up the local error   \eqref{eq: err-estimator}. This will be used as an error estimator. They are shown in
Fig. \ref{fig: err}. Similar to  the direct numerical experiment in
Fig. \ref{fig: chain}, the minimum of the error can be found around $N_d=70.$

\begin{figure}[htbp]
\begin{center}
\includegraphics[scale=0.205]{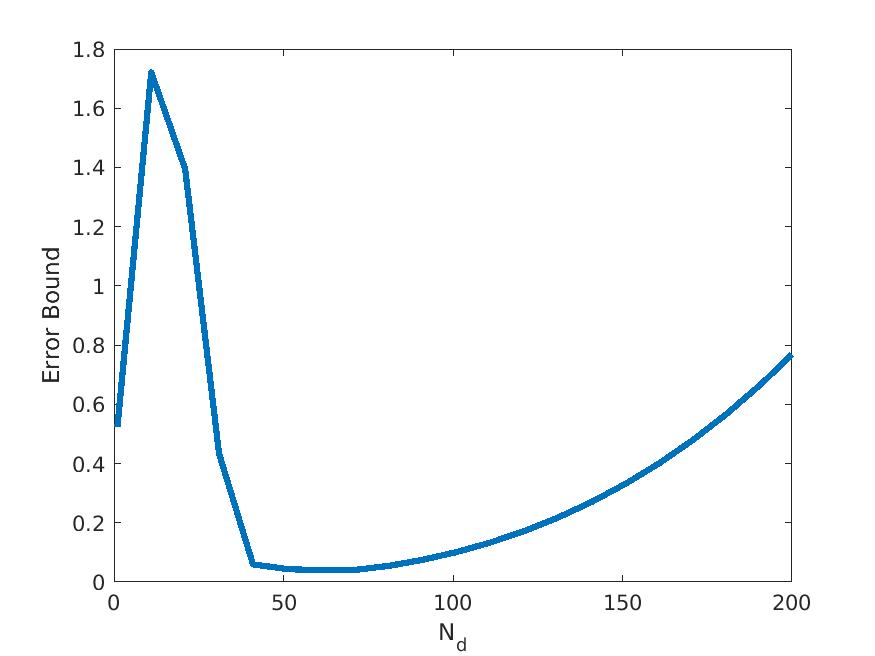}
\caption{The MSE  error estimator \eqref{eq: err-estimator} for the test case $n=12$  and $N_{Gate}= 2^{12}$. }
\label{fig: err}
\end{center}
\end{figure}

\medskip

It is clear that the simple splitting \eqref{eq: U0} for the deterministic group of Hamiltonians in $H_0$ \eqref{eq: H-0} can be replaced by higher order Trotter splitting methods. As an experiment, we replaced \eqref{eq: U0} by the symmetric-splitting method,
\begin{equation}\label{eq: U0-symm}
U_0(t) \approx \exp\big(-\frac{ith_1}2\big)  \cdots \exp \big(-\frac{ith_{N_d-1}}2\big)
\exp \big(- ith_{N_d}\big) \exp\big(-\frac{ith_{N_d-1}}2\big) \cdots \exp \big(-\frac{ith_{1}}2\big).
\end{equation}
In this case, we expect that the bias term \eqref{eq: err-estimator} will be reduced. Fig. \ref{fig: chain0-symm} shows the MSE from this experiment. In order to keep the gate number fixed, we choose the step size as, $\dt= \frac{2N_d T}{N_{Gate}}$, since at each step, \eqref{eq: U0-symm} involves $2N_d-1$ unitary operators. We observe by comparing to Fig. \ref{fig: chain} (left panel) that the choice $N_d=60$ yields the optimal accuracy, suggesting that the more accuracy treatment \eqref{eq: U0-symm} of $U_0$ tilts the balance \eqref{eq: err-estimator} more toward a deterministic method.   

\begin{figure}[htbp]
\begin{center}
\includegraphics[scale=0.15]{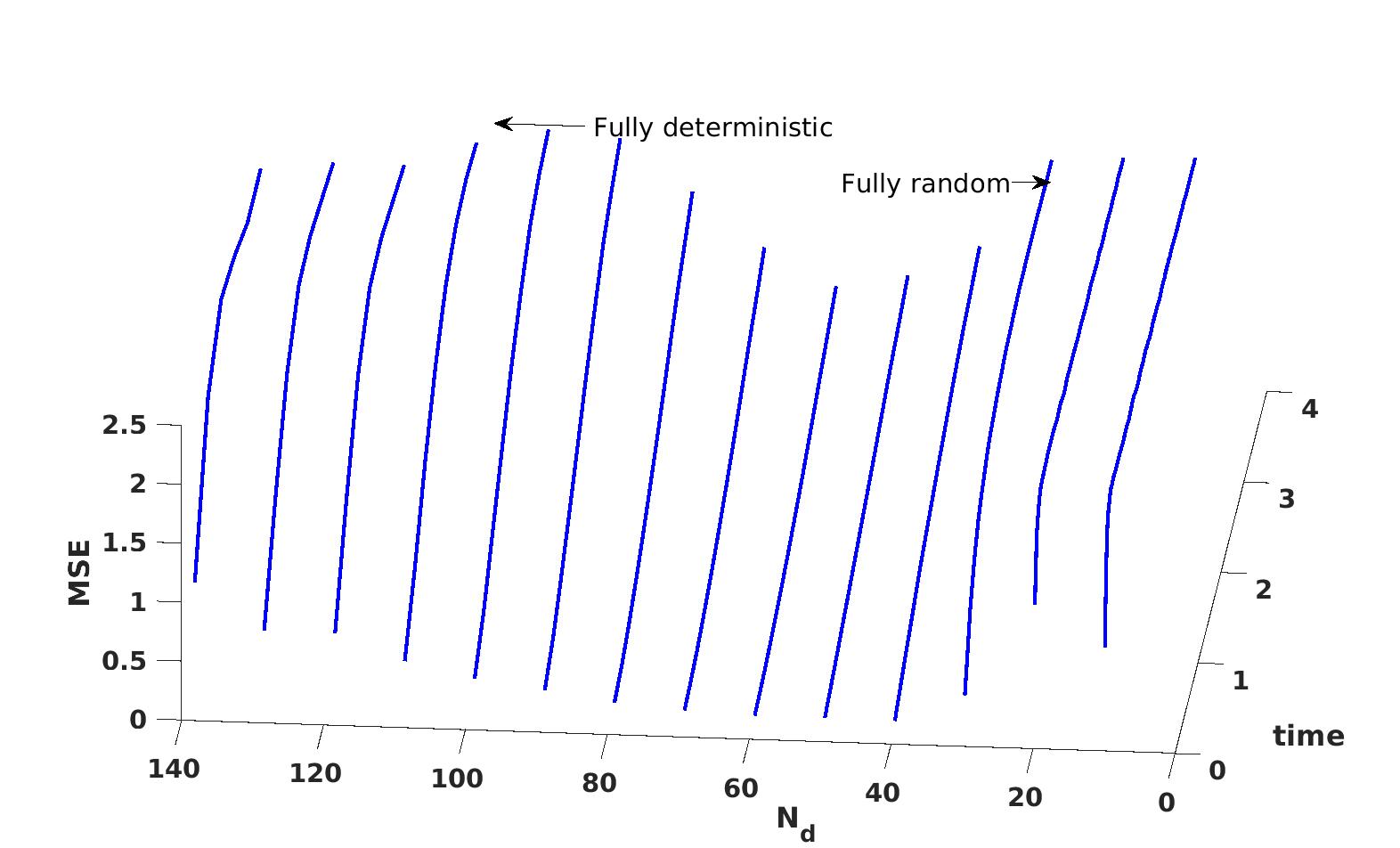}
\caption{The MSE for the hybrid algorithm applied to the model \eqref{eq: hchain} with $n=10$ and $N_{Gate}= 2^{11}$. The batch size $K=1.$  The deterministic part $U_0(t)$ is approximated by the symmetric-splitting method \eqref{eq: U0-symm}. }
\label{fig: chain0-symm}
\end{center}
\end{figure}

\medskip

In the hybrid algorithm, we have used the important sampling algorithm to pick up the random Hamiltonian terms. Here we also ran a test for the system  $n=12$ and $ N_{Gate}= 2^{12}$, where the random terms are sampled according to the uniform distribution ($p_j=1/N_r$). The MSE is plotted in Fig. \ref{fig: chain4}. Again, we observe that the hybrid method can achieve better accuracy, compared to the fully deterministic and fully random approaches.  Compared to the important sample algorithm, (Fig. \ref{fig: chain}, right panel), the MSE is slightly larger. Therefore, the important sampling approach should be preferred, whenever possible.
\begin{figure}[htbp]
\begin{center}
\includegraphics[scale=0.205]{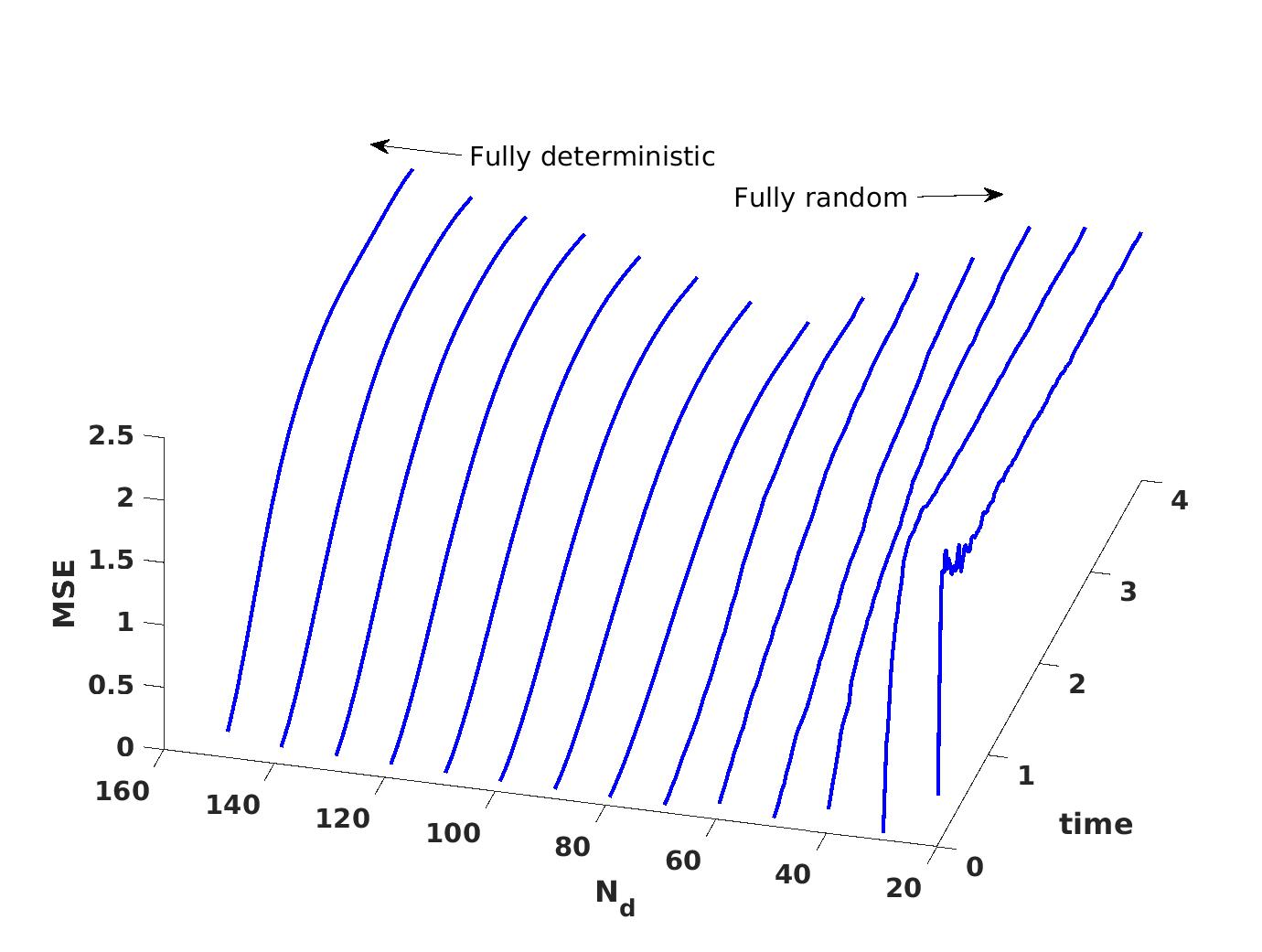}
\caption{The MSE for the hybrid algorithm, applied to the Heisenberg chain \eqref{eq: hchain} with $n=12$ and $N_{Gate}= 2^{12}$.  The random terms are selected with uniform probability.}
\label{fig: chain4}
\end{center}
\end{figure}

In all previous tests, we have chosen the batch size to be one, i.e., $K=1.$
Finally, we study the effect of using different batch sizes. For this purpose, we fix $n=10$ and $ N_{Gate}= 2^{11}$. Then we set $K=4$ and $K=8$. The resulting
MSE is shown in Fig. \ref{fig: chainK}. We observe that they have similar performance. 

\begin{figure}[htbp]
\begin{center}
\includegraphics[scale=0.125]{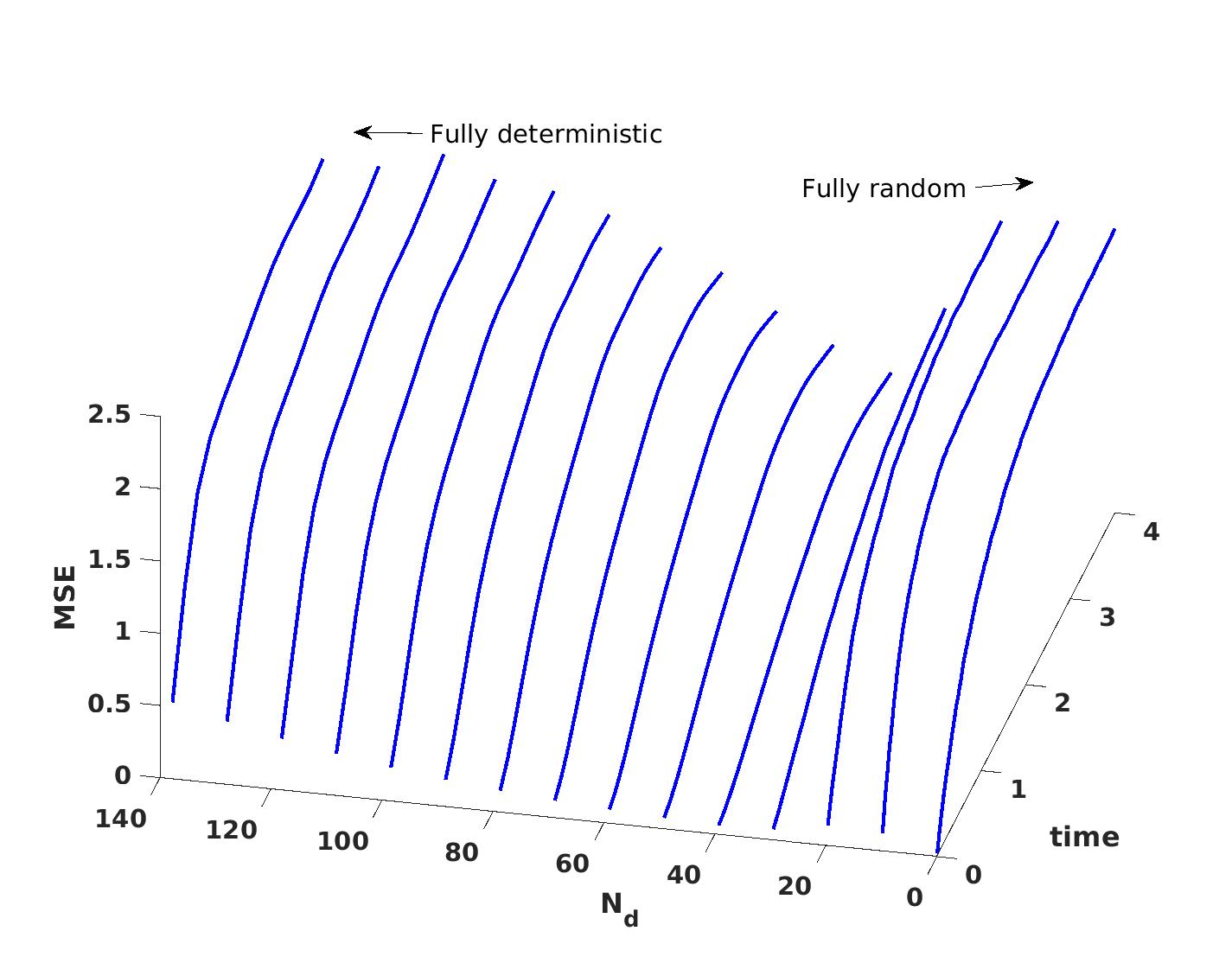}
\includegraphics[scale=0.1333]{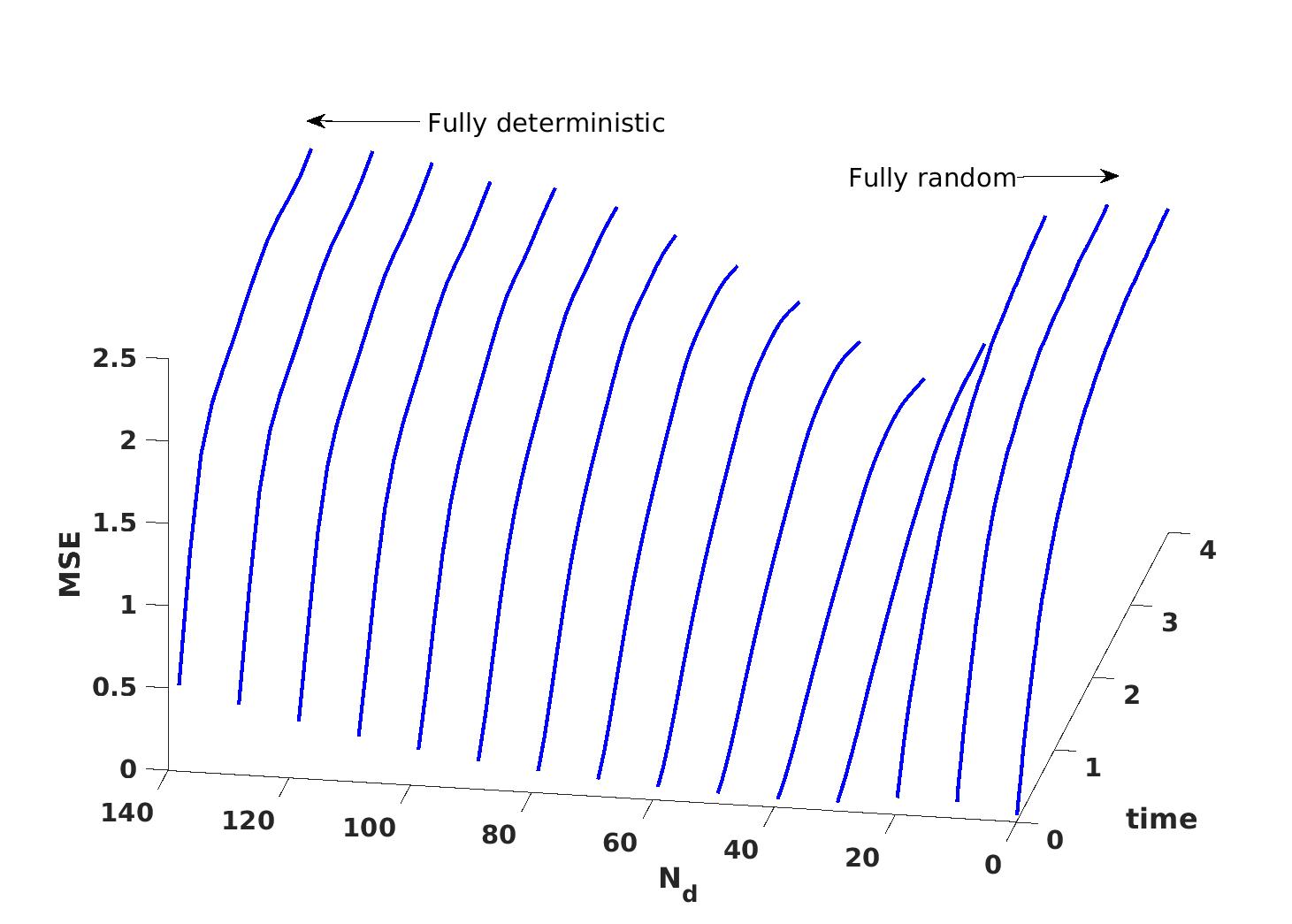}
\caption{The MSE for the hybrid algorithm applied to the Heisenberg chain \eqref{eq: hchain} with $n=10$ and  $N_{Gate}= 2^{11}$. Left: $K=4$; Right: $K=8$. }
\label{fig: chainK}
\end{center}
\end{figure}

\section{Summary}
In this paper, we proposed a Trotter approximation algorithm, where the total Hamiltonian is divided into two groups that will be treated in completely different
manners. The first group, consisting of Hamiltonian terms with larger amplitude, is implemented using standard Trotter splitting. For the terms in the other group, they are sampled randomly. We showed that the error from such a hybrid algorithm  has a typical bias-variance trade-off. Numerous examples suggest that by balancing the two types of error, one can obtain an optimal accuracy.

We have obtained the asymptotic form of \eqref{eq: err-estimator} with respective to the step size $\dt$. But the prefactors may not be available in advance. An empirical observation is that the optimal cut-off, i.e., $N_d$, is often where the coefficients of the Hamiltonians have a drop in magnitude. A more systematic approach to determine the partition is a remaining challenge.   

Another important question is how such error will be passed down to the following steps, e.g., phase estimation and solutions of linear systems. For example, it was shown in \cite{babbush2015chemical} that the typical error bounds of Trotter-Suzuki formulas can be overestimations  if only the properties of the ground state are desired. 

\appendix
\section{The proof of Proposition \ref{eq: mean-1step}}
\begin{proof}
Recall that $ \ket{\psi(t)}:= \exp  \big( -i t H  \big) \ket{\psi(0)} $ is the exact solution.
We will split the error as follows,
\begin{equation}
\begin{aligned}
  \left\| \ket{\psi_1}  - \mathbb{E}[\ket{\phi_1} ] \right\|  & \leq   \left\| \ket{\psi(\dt)}  -   \exp \big( -i \dt H_0\big) \exp \big( - i \dt H_1 \big)   \ket{\psi_0}   \right\| \\
 & + \left\| \exp \big( -i \dt H_0\big) \exp \big( - i \dt H_1 \big)   \ket{\psi_0}    - \mathbb{E}[\ket{\phi(\dt)} ] \right\|.
\end{aligned}
\end{equation}

The first term can be estimated by \eqref{eq: ebound'} from \cref{eq: split-error}. For the second term, it is enough to consider  the difference, denoted by $\ket{\theta(t)},$ of the following wave functions,
\begin{equation}
\ket{\theta(t)} = \exp \big( -i t H_1  \big)  |\psi_0\rangle - \mathbb{E}\left[ \exp \big( -i t H_1 - i t \delta\!H \big)\right]  |\psi_0\rangle
\end{equation}

Let $\ket{\phi(t)}= \big( -i t H_1 - i t \delta\!H \big)  |\psi_0\rangle$, and  notice that
   $$\dsp \frac{d}{dt} \mathbb{E}[\ket{\phi(\dt)} ] =  \mathbb{E}[   -i (H_1+\delta\!H)  \ket{\phi(\dt)} ].  $$   
 Therefore, the error term term  $\ket{\theta(t)}$ satisfies the equation,
  \begin{equation}
  \frac{d}{dt} \ket{ \theta(t)}  =  -i H_1  \ket{\theta(t)} -i \mathbb{E} [  \delta\!H  \ket{\phi(t)}].
\end{equation}

Using the variation-of-constant formula, and using the fact that $ \mathbb{E} [  \delta\!H]=0,$ we can write,
\[
\begin{aligned}
 \ket{\theta(\dt)} =& -i \int_0^{\dt}  U_1(\dt -t) \mathbb{E} \left[  \delta\!H  \ket{\phi(t)} \right] dt\\
		=  & -i \int_0^{\dt} U_1(\dt -t) \mathbb{E} \left[  \delta\!H ( \ket{\phi(t)} - \ket{\phi(0)} )  \right] dt\\
		= & - \int_0^{\dt} U_1(\dt -t) \mathbb{E} \left[  \delta\!H (H_1 + \dH) \int_0^t \ket{\phi(\tau)} \right] d\tau dt
\end{aligned}
\]

Here $U_1(t)=\exp( -i t H_1).$ In light of the fact that $\| \ket{\phi}\| =1,$ one can take the norms and arrive at,
\[
\|  \ket{\chi(\dt)}  \| \le  \int_0^{\dt} \int_0^t \mathbb{E} \left[  \|\delta\!H (H_1 + \dH)\| \right]^{1/2} d\tau dt,\]
which yields the second term in the estimate \eqref{eq: err-mean}.

\end{proof}

\section*{Acknowlegement}
Jin's research was supported by Innovation Program of Shanghai Municipal Education Commission (No. 2021-01-07-00-02-E00087). Li's research has been supported by NSF DMS-1953120 and DMS-2111221.

\bibliographystyle{plain}
\bibliography{qcomp,errorestimate,trotter,rbm}

\end{document}